\documentclass{amsart}
\usepackage[T1]{fontenc}
\usepackage{graphicx}
\usepackage{amsmath,amsfonts,amssymb,amscd}
\usepackage{physics}
\usepackage{dsfont}
\usepackage[margin=1in]{geometry}
\usepackage[pagebackref, colorlinks = true, linkcolor = blue, urlcolor  = blue, citecolor = red]{hyperref}
\usepackage{comment}
\usepackage[capitalise]{cleveref}
\usepackage[table,xcdraw]{xcolor}
\usepackage{textcomp}

\newtheorem{theorem}{Theorem}[section]
\newtheorem{lemma}[theorem]{Lemma}
\newtheorem{proposition}[theorem]{Proposition}
\newtheorem{definition}[theorem]{Definition}
\newtheorem{corollary}[theorem]{Corollary}
\newtheorem{conjecture}[theorem]{Conjecture}
\newtheorem{remark}[theorem]{Remark}
\newtheorem*{theorem*}{Theorem}

\newcommand{\E}{\mathbf{E}}

\renewcommand{\P}{\mathbf{P}}
\newcommand{\N}{\mathbb{N}}
\newcommand{\C}{\mathbb{C}}

\renewcommand{\epsilon}{\varepsilon}

\title{Random measurements are almost maximally incompatible}

\author{Andreas Bluhm}
\address[Andreas Bluhm]{Univ.~Grenoble Alpes, CNRS, Grenoble INP, LIG, France}
\email{andreas.bluhm@univ-grenoble-alpes.fr}
\author{C{\'e}cilia Lancien}
\address[Cécilia Lancien]{Univ.~Grenoble Alpes, CNRS, Institut Fourier, France}
\email{cecilia.lancien@univ-grenoble-alpes.fr}
\author{Ion Nechita}
\address[Ion Nechita]{Univ.~Toulouse, CNRS, Laboratoire de Physique Th\'eorique, France}
\email{nechita@irsamc.ups-tlse.fr}

\date{\today}


\begin{document}

\begin{abstract}
    In this work, we investigate the incompatibility of random quantum measurements. Most previous work has focused on characterizing the maximal amount of white noise that any fixed number of incompatible measurements with a fixed number of outcomes in a fixed dimension can tolerate before becoming compatible. This can be used to quantify the \emph{maximal} amount of incompatibility available in such systems. The present article investigates the incompatibility of several classes of \emph{random} measurements, i.e., the \emph{generic} amount of incompatibility available. In particular, we show that for an appropriate choice of parameters, both random dichotomic projective measurements and random basis measurements are close to being maximally incompatible. We use the technique of incompatibility witnesses to certify incompatibility and combine it with tools from random matrices and free probability.
\end{abstract}

\maketitle

\tableofcontents

\newpage 

\section{Introduction and main results}

\subsection{Introduction}
One of the prime examples of how quantum mechanics differs from classical mechanics is the existence of incompatible measurements. A set of measurements is incompatible if the outcomes of the measurements do not arise as classical post-processings of some joint measurement. Equivalently, measurements are incompatible if they do not arise as marginals from a joint measurement \cite{Heinosaari2016}. In this sense, they cannot be measured simultaneously \cite{Heisenberg1927, Bohr1928}. 
The position and momentum measurement are arguably the best-known examples of incompatible measurements, but interesting situations also occur in finite dimensions \cite{Heinosaari2016,guhne2021incompatible}. For example, two projective measurements are compatible if and only if their elements commute. The existence of incompatible measurements is crucial for the violation of Bell inequalities \cite{Fine1982} and therefore enables many quantum information processing protocols \cite{Brunner2014}, for example device-independent quantum key distribution \cite{pironio2009device}. Therefore, incompatibility plays a similarly fundamental role for quantum technologies as entanglement, and can hence be seen as a resource for quantum information tasks  \cite{Heinosaari2015}. 

It is well known that entanglement vanishes given enough noise, which is part of what makes building quantum computers very challenging. The same is true for measurement incompatibility  \cite{busch2013comparing}. It is therefore important to quantify how much noise a given set of measurements can tolerate before the measurements become compatible. Most of the work so far has focused either on concrete sets of measurements \cite{guhne2021incompatible}, for example qubit measurements \cite{busch1986unsharp, busch2008approximate, grinko2024compatibility} or mutually unbiased bases \cite{Carmeli2012, Designolle2018}, or on the maximal noise robustness of incompatibility for a given set of physical parameters, such as the dimension of the quantum system, the number of measurements, or the number of outcomes. The latter question has been investigated in \cite{busch2013comparing, Gudder2013} and recently also in  \cite{bluhm2018joint, bluhm2020compatibility, bluhm2022GPT, bluhm2022tensor}, where some of the present authors discovered a connection to matrix convex sets and tensor norms.

Less explored is the question of quantifying the incompatibility of generic or typical quantum measurements, i.e., understanding how incompatible random quantum measurements are. Different models of random measurements have been defined in \cite{heinosaari2020random}, starting from random completely positive maps \cite{collins2016random}. In that paper, the authors in particular compared different compatibility criteria for random measurements. The article \cite{zhang2019incompatibility} calculated the incompatibility of pairs of random qubit measurements.

The present article continues this line of work by considering special classes of random measurements, in particular random dichotomic projective measurements and random basis measurements. Our leading question is to determine whether such measurements are close to being maximally incompatible. We use the technique of incompatibility witnesses to certify incompatibility and combine it with tools from random matrices and free probability. Analogously to entanglement witnesses, incompatibility witnesses are hyperplanes that separate certain incompatible measurements from the convex set of compatible measurements.

\subsection{Main results}

To quantify the amount of incompatibility between measurements, we use the \emph{minimal compatibility degree} $\tau(d, g, (k_1, \ldots, k_g))$, i.e.,  $1- \tau(d, g, (k_1, \ldots, k_g))$ is the minimal amount of white noise needed to make any set of $g$ measurements in dimension $d$ with $k_x$ outcomes for the $x$-th measurement compatible (see Definition \ref{def:max-incmomp-degree}).  For the compatibility degree of a fixed set of measurements $\{\mathrm M_1,\ldots,\mathrm M_g\}$, we write $\tau(\mathrm M_1,\ldots,\mathrm M_g)$ (see Definition \ref{def:incmomp-degree}). 

Our first results are about dichotomic measurements (i.e., $k_x=2$ for all $x$). For those the minimal compatibility degree can be computed exactly, at least in the regime where the dimension is exponentially larger than the number of measurements. Concretely, it has been shown in \cite{bluhm2018joint} that
\begin{equation} \label{eq:tau-2-lower}
    \tau(d, g, (2, \ldots, 2)) \geq \frac{1}{\sqrt{g}},
\end{equation}
with equality for $d = \Omega(2^g)$.

We first look at the case where $g=2$. We consider $E,F\subset\C^d$ to be independent uniformly distributed random subspaces of respective dimensions $\alpha d,\beta d$, for some fixed parameters $0<\alpha,\beta<1$. We denote by $P_E$ and $P_F$ the orthogonal projections onto these subspaces and by $\mathrm P_E = (P_E, I-P_E)$ and $\mathrm P_F = (P_F, I-P_F)$ the dichotomic measurements defined by these projections, where $I$ is the identity matrix. We find that, for $\alpha$ and $\beta$ close enough to $1/2$, these measurements are close to being maximally incompatible. The following is an informal restatement of Corollary \ref{cor:two-proj-max-incomp}:

\begin{theorem*}
Suppose that $(\alpha-\frac{1}{2})^2+(\beta-\frac{1}{2})^2 \leq \frac{1}{4}$. Then, 
\begin{equation*}
    \tau(\mathrm P_E, \mathrm P_F) \xrightarrow[d \to \infty]{} \frac{1}{\sqrt{2}}.
\end{equation*}
\end{theorem*}
For $\alpha$ and $\beta$ not close enough to $1/2$, the behavior changes and the measurements become less and less incompatible. For instance, the following is an informal restatement of Corollaries \ref{cor:upper-two-proj-unbalanced} and \ref{cor:lower-two-proj-unbalanced}:
\begin{theorem*}
    Suppose that $\beta = \alpha$ and that either $ \alpha > \frac{1}{2}(1+\frac{1}{\sqrt{2}})$ or $\alpha < \frac{1}{2}(1-\frac{1}{\sqrt{2}})$. Then, setting $\lambda_\alpha=4\alpha(1-\alpha)$,
    \begin{equation*}
        \tau(\mathrm P_E,\mathrm P_F) \xrightarrow[d \to \infty]{} \frac{1}{\sqrt{\lambda_\alpha}+\sqrt{1-\lambda_\alpha}}.
    \end{equation*} 
\end{theorem*}
    
Turning to the case where $g>2$, we consider $g$ independent uniformly distributed random projections $P_x$ of rank $d/2$ on $\mathbb C^d$. We write $\mathrm P_x=(P_x,I-P_x)$ for the associated dichotomic measurements. Then, we find that for $g$ large (but scaling at most linearly with $d$), this set of measurements is close to being maximally incompatible. The following is a consequence of Theorem \ref{th:deviation-final}:
\begin{theorem*}
    Let $g \leq 3d$. Then, with probability at least $1-e^{-2d}$, 
    \begin{equation*}
        \tau(\mathrm P_1, \ldots, \mathrm P_g) \leq \frac{31}{\sqrt{g}}.
    \end{equation*}
\end{theorem*}

Combined with the lower bound recalled in Eq.~\eqref{eq:tau-2-lower}, this result implies that 
\begin{equation} \label{eq:tau-2}
    \tau(d,g,(2, \ldots, 2)) = \Theta\left(\frac{1}{\sqrt{g}}\right) \qquad \forall d \geq \frac{g}{3} ,
\end{equation}
while previously we only knew this to hold for $d\geq 2^{\lceil (g-1)/2 \rceil}$.

Our next results concern a set of $g$ independent uniformly distributed random basis measurements $\mathrm B_x$ in dimension $d$. As these measurements are not dichotomic, we are no longer comparing the value of $\tau(\mathrm B_1,\ldots,\mathrm B_g)$ to $1/\sqrt{g}$. In fact, little is known about the compatibility degree of measurements which are not dichotomic. In terms of lower bounds, there is the universal lower bound that we can obtain from approximate cloning \cite{werner1998optimal}, which in fact holds for any $g$ basis measurements with $d$ outcomes in dimension $d$, namely:
\begin{equation} \label{eq:tau-d-lower}
    \tau(\mathrm B_1, \ldots, \mathrm B_g) \geq \frac{g + d}{g(d+1)} .
\end{equation}

For $g=2$, it is in fact easy to show that two independent random basis measurements are asymptotically maximally incompatible. The following is a consequence of Corollary \ref{cor:two-bases}:
\begin{theorem*}
    It holds that $\tau(\mathrm B_1,\mathrm B_2) \xrightarrow[d \to \infty]{} \frac{1}{2}$.
\end{theorem*}

In fact, we find that random basis measurements are also close to being maximally incompatible for $g$ large enough. The following is a consequence of Theorem \ref{th:deviation-final-bases}:
\begin{theorem*}
    Let $g \geq 2d$. Then, with probability at least $1-e^{-2g\log(d)}$,
    \begin{equation*}
        \tau(\mathrm B_1, \ldots, \mathrm B_g) \leq \frac{135 \log(d)-1}{d-1}.
    \end{equation*}
\end{theorem*}

We have thus shown that the lower bound from approximate cloning recalled in Eq.~\eqref{eq:tau-d-lower} is asymptotically tight, up to log factors, in the regime where $g$ is of order at least $d$: 
\begin{equation} \label{eq:tau-d}
    \tau(d,g,(d, \ldots, d)) = \tilde\Theta\left(\frac{1}{d}\right) \qquad \forall g \geq 2d.
\end{equation}

We can compare Eq.~\eqref{eq:tau-2} and Eq.~\eqref{eq:tau-d} when $g$ and $d$ are of the same order. In this regime, we see that the minimal compatibility degree of $2$-outcome measurements is of order $1/\sqrt{g}$ while that of $d$-outcome measurements is of order $1/g$ (up to log factors).

Finally, we also analyze random induced POVMs, in the sense of \cite{heinosaari2020random}. We show that, at least in some regime, their compatibility degree approaches the maximal value as well (up to multiplicative constants that do not depend on the involved parameters). 

In conclusion, our results show indeed that in the appropriate regimes, the different models of random measurements that we consider in this article are close to being maximally incompatible. This supports the idea that being very incompatible is a generic property of quantum measurements on high-dimensional quantum systems.

\subsection{Organization of the paper}
The present paper is structured as follows: In Section \ref{sec:prelim}, we review some findings concerning measurement incompatibility, incompatibility witnesses, and random matrices. In Section \ref{sec:two-projections}, we focus on the case of two random dichotomic projective measurements, before considering more random dichotomic projective measurements in Section \ref{sec:more-projections}. Switching the type of considered random measurements, we study in Section \ref{sec:random-bases} random basis measurements, before turning to general random measurements in Section \ref{sec:random-POVMs}. We conclude with a discussion of our findings and some open questions in Section \ref{sec:discussion}.

\section{Preliminaries} \label{sec:prelim}

\subsection{Notation}
For $d\in \mathbb N$, we will write as a shorthand $[d]:=\{1, \ldots, d \}$. We write $M_d(\mathbb C)$ for the set of $d \times d$ matrices with complex entries and $M_d^{\mathrm{sa}}(\mathbb C)$ for the subset of Hermitian such matrices. We write $I_d$ for the identity matrix in dimension $d$, where sometimes we will drop the index indicating the dimension. Given $A \in M_d^{\mathrm{sa}}(\mathbb C)$, we denote by $\lambda_{\max}(A)$, resp.~$\lambda_{\min}(A)$ the maximal, resp.~minimal, eigenvalue of $A$. We write $\|A\|_p$ for the Schatten $p$-norm of $A$ for $p \in [1, \infty]$ (in particular $\|A\|_\infty$ stands for its operator norm). Given $v\in\C^d$, we write $\|v\|_p$ for the $\ell_p$-norm of $v$ for $p \in [1, \infty]$, and use the lighter notation $\|v\|$ instead of $\|v\|_2$ for its Euclidean norm.

\subsection{Measurement incompatibility}

In this section, we will review some background on measurement incompatibility.  We refer the reader to \cite{Heinosaari2011} for an introduction to the mathematical description of quantum systems and to \cite{Heinosaari2016} for a general introduction to measurement incompatibility. For quantum systems of dimension $d \in \mathbb N$, the set of quantum states is given as
\begin{equation*}
    S_d(\mathbb C) = \left\{\rho \in M_d(\mathbb C) : \rho \geq 0, \operatorname{Tr}[\rho] = 1\right\}.
\end{equation*}
Measurements map quantum states to probability distributions. Therefore, a measurement with outcomes in a finite set $\Sigma$ is a tuple of matrices $(E_i )_{i \in \Sigma} \subseteq M_d(\mathbb C)$ such that
\begin{enumerate}
    \item $E_i \geq 0$ for all $i \in \Sigma$,
    \item $\sum_{i \in \Sigma} E_i = I_d$.
\end{enumerate}
Sets of matrices satisfying (1) and (2) are called \emph{positive operator-valued measures (POVMs)}, and each $E_i$ is referred to as an \emph{effect} of the POVM $(E_i )_{i \in \Sigma}$. Labeling the outcomes using natural numbers, we can always take $\Sigma = [k]$ for some $k \in \mathbb N$. If all elements of a POVM are projections, the POVM is a \emph{projection-valued measurement (PVM)}, also simply called \emph{projective measurement}. 

Now, we consider a set of measurements, i.e.,  $(E_{i|x})_{i \in [k_x]}$ is a POVM for every $x \in [g]$. Here, $g \in \mathbb N$ is the number of measurements and $k_x \in \mathbb N$ is the number of outcomes for the measurement with label $x$. 

An important example of a set of measurements are those obtained from \emph{mutually unbiased bases (MUBs)}: For every $x \in [g]$, let $\{v_{i|x}\}_{i \in [d]}$ be an orthonormal basis of $\mathbb C^d$. Then, we can construct POVMs from this by setting $E_{i|x} = |v_{i|x}\rangle\!\langle v_{i|x}|$ for all $x \in [g]$, $i \in [d]$. These bases are called \emph{mutually unbiased} if 
\begin{equation*}
    \left|\langle v_{i|x}|v_{i'|x'}\rangle\right| = \frac{1}{\sqrt{d}}
\end{equation*}
for all $x,x' \in [g]$ such that $x \neq x'$ and all $i,i' \in [d]$. It is known that there are at most $g \leq d+1$ MUBs, but it is unknown whether this upper bound can always be reached. We refer the reader to \cite{durt2010mutually} for more background on MUBs.

\begin{definition}[Compatible measurements]
    Let $g \in \mathbb N$, $d \in \mathbb N$, and $k_x \in \mathbb N$ for all $x \in [g]$. Let $\{ (E_{i|x})_{i \in [k_x]} \}_{x \in [g]}$ be a set of $g$ $d$-dimensional POVMs. These measurements are \emph{compatible} if there exists another $d$-dimensional POVM $(J_{i_1, \ldots, i_g})_{i_1 \in [k_1], \ldots, i_g \in [k_g]}$ such that 
    \begin{equation*}
        E_{i|x} = \sum_{\substack{y\in[g] \\ y\neq x}}\sum_{i_y \in [k_y]} J_{i_1, \ldots, i_{x-1}, i, i_{x+1}, \ldots, i_{g}} \qquad \forall i \in [k_x], \forall x \in [g] \, .
    \end{equation*}
    If the measurements are not compatible, they are called \emph{incompatible}.
\end{definition}
Equivalently, the measurements $\{ (E_{i|x})_{i \in [k_x]} \}_{x \in [g]}$ are compatible if there exist a finite set $\Lambda$, a POVM $(M_\lambda)_{\lambda \in \Lambda}$ and conditional probability distributions $(p_{i|\lambda, x})_{i \in [k_x]}$ for all $x \in [g]$ and $\lambda \in \Lambda$ such that 
\begin{equation*}
    E_{i|x} = \sum_{\lambda \in \Lambda} p_{i|\lambda, x} M_\lambda \qquad \forall i \in [k_x], \forall x \in [g] \, .
\end{equation*}
While the latter definition looks more complicated, it has an operational interpretation: the measurements $(E_{i|x})_{i \in [k_x]}$ are compatible if we can reproduce their outcome statistics by performing a single measurement $(M_\lambda)_{\lambda \in \Lambda}$ and by subsequently post-processing the outcomes according to the probability distributions $(p_{i|\lambda, x})_{i\in[k_x]}$. 

For the special case of projective measurements, compatibility is equivalent to commutativity: two projective measurements are compatible if and only if all their effects commute pairwise. However, for general POVMs, the situation is more complicated. In this case, compatibility cannot be characterized by commutativity alone, and one typically resorts to semidefinite programming (SDP) techniques to test for compatibility \cite{boyd2004convex}. There exist various sufficient criteria for compatibility. For example, the \emph{Jordan product criterion} (see \cref{lem:Jordan}) provides a practical sufficient condition for a set of POVMs to be compatible.

It is well known that noise eventually destroys entanglement, which makes building a quantum computer very challenging. In the same manner, incompatible measurements can be made compatible by adding a sufficient amount of noise. We will focus on uniform noise, also called white noise. We refer the reader to \cite{designolle2019incompatibility} for an overview and a comparison of different noise models.
\begin{definition}\label{def:noise-model}
    Let $k \in \mathbb N$ and let $(E_i)_{i \in [k]}$ be a POVM. Let $t \in [0,1]$ be a noise parameter. Then, we define the noisy POVM $(E_i^{(t)})_{i \in [k]}$ as
    \begin{equation*}
        E_i^{(t)} := t E_i + (1-t)\frac{I}{k} \qquad \forall i\in[k].
    \end{equation*}
    Note that $E_i^{(1)} = E_i$ and that $E_i^{(0)} = I/k$.
\end{definition}
Given a set of measurements, we can define its compatibility degree as the minimal amount of noise which makes them compatible. In a way, the compatibility degree thus quantifies the amount of incompatibility present in these measurements.
\begin{definition}[Compatibility degree of a set of measurements] \label{def:incmomp-degree}
Let $d$, $g \in \mathbb N$ and $k_x \in \mathbb N$ for all $x \in [g]$. Let $E:=\{ (E_{i|x})_{i \in [k_x]} \}_{x \in [g]}$ be a set of POVMs. Then, their \emph{compatibility degree} is defined as  
    \begin{equation*}
        \tau(E) := 
        \max\left\{ t \in [0,1]: \left\{\left(E_{i|x}^{(t)}\right)_{i \in [k_x]}\right\}_{x \in [g]}\text{ is a set of compatible measurements} \right\} \, .
    \end{equation*}
\end{definition}
Fixing the dimension of the quantum system, the number of measurements we want to consider and their outcomes, we can also ask for the minimal amount of noise that makes any measurements fitting these criteria compatible.
\begin{definition}[Minimal compatibility degree] \label{def:max-incmomp-degree}
Let $d$, $g \in \mathbb N$ and $k_x \in \mathbb N$ for all $x \in [g]$. Then, the \emph{minimal compatibility degree} is defined as  
    \begin{align*}
        \tau(d,g,(k_1, \ldots, k_g)) := \min_E \tau(E) \,,
    \end{align*}
    where the minimum is taken over all sets $E$ of $g$ POVMs in dimension $d$, where the $x$-th POVM has $k_x$ outcomes for all $x \in [g]$.
\end{definition}

In the following, we collect what is known about the compatibility degree.
\begin{proposition} \label{prop:bounds-incompat-degree}
    Let $d$, $g \in \mathbb N$ and $k_x \in \mathbb N$ for all $x \in [g]$. Let $E:=\{ (E_{i|x})_{i \in [k_x]} \}_{x \in [g]}$ be a set of POVMs and $M:=\{ (M_{i|x})_{i \in [d]} \}_{x \in [g]}$ be a set of MUBs. Then,
    \begin{enumerate}
        \item $\tau(d,g, (2, \ldots, 2)) \geq \frac{1}{\sqrt{g}}$.
        \item $\tau(d,g, (2, \ldots, 2)) = \frac{1}{\sqrt{g}}$ whenever $d \geq 2^{\lceil (g-1)/2\rceil}$.
        \item $\tau(d,g, (2, \ldots, 2)) \geq c(d)$ with $c(d) = 4^{-d'}\binom{2d'}{d'}$ for $d'=\lfloor d/2 \rfloor$ (so that $c(d)\sim\sqrt{\smash[b]{(2/\pi)/d}}$ as $d\to\infty$).
        \item $\tau(E) \geq \frac{1}{g}$.
        \item $\tau(E)\geq \frac{g+k_{\max}d}{g(1+k_{\max}d)}$ for $k_{\max} = \max_{x \in [g]} k_x$.
        \item $\tau(M)\leq \frac{g + \sqrt{d}}{g(\sqrt{d} + 1)}$.
    \end{enumerate}
\end{proposition}
\begin{proof}
    The first and second points were proved in \cite{bluhm2018joint}, based on results by \cite{passer2018minimal}. The third point was proved in \cite{bluhm2022steering}, using ideas from \cite{ben-tal2002tractable}. The fourth point can be found in \cite{Heinosaari2016}. The fifth point was proved in \cite[Proposition 6.7]{bluhm2020compatibility}, using the idea of \cite{Heinosaari2014} that approximate cloning can be used to find a joint measurement. The sixth point was proved in \cite{Carmeli2012, Designolle2018}.
\end{proof}

\subsection{Incompatibility witnesses}\label{sec:incompatibility-witnesses}
Later on, we will be using the method of incompatibility witnesses to prove that the random measurements we consider are incompatible. Incompatibility witnesses were introduced in \cite{jencova2018incompatible} and \cite{carmeli2019quantum} and later studied in \cite{bluhm2020compatibility, bluhm2022GPT, bluhm2022tensor} from a slightly different perspective. As for entanglement witnesses, incompatibility witnesses are hyperplanes separating some incompatible measurements from the compact convex set of compatible measurements. 

Let $(E_{i|x})_{i \in [k_x]}$ be a quantum measurement for all $x \in [g]$. \begin{definition}
    Let $d,g \in \mathbb N$ and $k_x \in \mathbb N$ for all $x \in [g]$. An incompatibility witness is a $g$-tuple of operators $W = (W_{i|x})_{i \in [k_x], x \in [g]}$, where $W_{i|x} \in M^{\mathrm{sa}}_d(\mathbb C)$ for all $i \in [k_x]$, $x \in[g]$, with the property that 
    $$\max_{E \text{ compatible}} \langle W, E \rangle \leq 1.$$
    Here, we have defined 
    \begin{equation*}
         \langle W, E \rangle := \sum_{x \in [g]}\sum_{i \in [k_x]} \operatorname{Tr}\left(W_{i|x} E_{i|x}\right)
    \end{equation*}
    and took the maximum over all sets of $g$ compatible POVMs in dimension $d$, where the $x$-th POVM has $k_x$ outcomes for $x \in [g]$.
\end{definition}

\begin{proposition} \label{prop:witness-bases}
    Let $d$, $g$, $k\in \mathbb N$ and $W = (W_{i|x})_{i \in [k], x \in [g]}$, where $W_{i|x} \in M^{\mathrm{sa}}_d(\mathbb C)$ for all $i \in [k]$, $x \in[g]$. Moreover, suppose that
    $$\max_{f:[g] \to [k]} \lambda_{\max}(W_f) \leq \frac{1}{d},$$
    where, for a function $f:[g] \to [k]$, we have defined 
    $$W_f:=\sum_{x \in [g]} W_{f(x)|x}.$$
    Then, $W$ is an incompatibility witness.
\end{proposition}
\begin{proof}
    Compatible measurements $E=(E_{\cdot | x})_{x\in[g]}$ can always be obtained by marginalizing a joint measurement $J$ indexed by functions $f:[g] \to [k]$: 
    $$E_{i|x} = \sum_{f:[g] \to [k] \,:\, f(x) = i} J_f \qquad \forall i \in [k], \forall x \in [g] .$$
    We thus have
    \begin{align*}
        \max_{E \text{ compatible}} \langle W,E \rangle &= \max_{E \text{ compatible}} \sum_{x\in[g]}\sum_{i\in[k]} \Tr\left(W_{i|x}E_{i|x}\right) \\
        &= \max_{J \text{ POVM}} \sum_{x\in[g]}\sum_{i\in[k]} \Tr\left(W_{i|x}\left( \sum_{f:[g] \to [k] \,:\, f(x)=i}J_f\right)\right) \\
        &= \max_{J \text{ POVM}} \sum_{f:[g] \to [k]} \Tr\left(\left(\sum_{x\in[g]} W_{f(x)|x}\right) J_f \right) \\     
        &= \max_{J \text{ POVM}} \sum_{f:[g] \to [k]} \Tr\left(W_fJ_f\right) \\  
        &\leq \max_{J_f \geq 0, \, \sum_f \Tr(J_f)=d} \sum_{f:[g] \to [k]}  \Tr\left(W_fJ_f\right) \\ 
        &= d\times \max_{\hat J \geq 0, \, \Tr(\hat J)=1} \Tr\big(\hat W \hat J\big) \\ 
        &= d\times \lambda_{\max}\big(\hat W\big) \\
        &= d\times \max_{f:[g] \to [k]} \lambda_{\max}(W_f), 
    \end{align*}
    where we have introduced the block-diagonal operators 
    $$\hat W = \sum_{f:[g] \to [k]} W_f \otimes \ketbra{f}{f} \quad \text{and} \quad \hat J = \sum_{f:[g] \to [k]} J_f \otimes \ketbra{f}{f}.$$
    And this proves exactly the announced result.
\end{proof}
For dichotomic measurements, the set of incompatibility witnesses has been characterized in \cite{bluhm2022GPT, bluhm2022tensor}. As dichotomic measurements are of the form $(E, I-E)$ for $0\leq E\leq I$, we can identify them with $E$ due to normalization. Thus, given a set $E:=\{E_x\}_{x \in [g]}$ where $E_x \in M_d(\mathbb C)$ such that $0\leq E_x\leq I$ for all $x \in [g]$, i.e., a set of $g$ dichotomic measurements in dimension $d$, we can simplify our definition of incompatibility witnesses in a slight abuse of terminology. We call a \emph{dichotomic incompatibility witness} a $g$-tuple of operators $W=(W_x)_{x \in [g]}$ where $W_x \in M_d^{\mathrm{sa}}(\mathbb C)$ for all $x \in [g]$ such that for $A=(2E_x-I_d)_{x \in [g]}$,
\begin{equation*}
\max_{E \text{ compatible}} \langle W, A \rangle \leq 1 \,.
\end{equation*}
From \cite[Eq.~(8)]{bluhm2022tensor}, we infer that $W$ is an incompatibility witness if and only if there exists a quantum state $\rho \in S_d(\mathbb C)$ such that
\begin{equation} \label{eq:incompatibility-witness-dichotomic}
    \rho - \sum_{x\in[g]} \epsilon_x W_x \geq 0 \qquad \forall \epsilon \in \{\pm 1\}^g.
\end{equation}
Incompatibility witnesses have the advantage that they can be found using semidefinite programs (SDPs). We refer to \cite{boyd2004convex} for an introduction to semidefinite programming. For simplicity, we consider the dichotomic case but the general case will just have different constraints on the $W_x$. We are given $0 \leq E_x \leq I_d$ and set $A_x = 2 E_x - I$ for all $x \in [g]$. Then, we compute 
\begin{align*}
    \lambda = \textrm{maximize} \quad & \sum_{x \in [g]} \operatorname{Tr}(A_xW_x) \\
    \textrm{subject~to} \quad & \sum_{x \in [g]} \epsilon_x W_x \leq \rho \qquad \forall \epsilon \in \{\pm 1\}^g \\ 
    \quad & \rho \in M^{\mathrm{sa}}_d(\mathbb C) \\
    \quad & \rho \geq 0 \\ 
    \quad & \!\Tr(\rho) = 1 \\ 
    \quad & W_x \in M^{\mathrm{sa}}_d(\mathbb C) \qquad \forall x \in [g].
\end{align*}
This is an SDP with variables $\rho$ and $\{W_x\}_{x \in [g]}$. If the optimal value $\lambda$ is strictly larger than $1$, the measurements $\{(E_x, I-E_x)\}_{x\in[g]}$ are incompatible and the $\{W_x\}_{x \in [g]}$ are an incompatibility witness. This SDP is essentially the dual SDP to the one presented in \cite{wolf2009measurements} to compute whether a given set of measurements is incompatible.

\subsection{Models of random measurements}\label{sec:random-matrix-models}

We discuss in this section the models of random matrices we are using in this work to address the question of (in-)compatibility of random measurements. The models of random quantum measurements, on the Hilbert space $\C^d$, that we are considering can be summarized in the following table. 

\bigskip

\bgroup
\def\arraystretch{1.2}
\begin{center}
\begin{tabular}{|c|c|c|c|}
\hline
\rowcolor[HTML]{EFEFEF} 
\textbf{\begin{tabular}[c]{@{}c@{}}Random POVM \\ model\end{tabular}} & \textbf{\begin{tabular}[c]{@{}c@{}}Number of \\ outcomes\end{tabular}} & \textbf{Parameters}         & \textbf{\begin{tabular}[c]{@{}c@{}}Relevant results\end{tabular}} \\ \hline
\begin{tabular}[c]{@{}c@{}}Dichotomic PVM\end{tabular}                      & $2$                                                                     & rank of the projections     &   \cref{sec:two-projections,sec:more-projections}                                                                  \\ \hline
\begin{tabular}[c]{@{}c@{}}Basis PVM\end{tabular}                   & $d$                                                                     & ---                           &  \cref{sec:random-bases}                                                                   \\ \hline
Induced POVM                                                               & $k$                                                                     & dimension of the ancilla space &                       \cref{sec:random-POVMs}                                              \\ \hline
\end{tabular}
\end{center}
\egroup

\bigskip

Before describing each model of random measurements in detail, let us point out that the general approach we are taking in this work is as follows: 
\begin{enumerate}
    \item Sample a $g$-tuple of independent random quantum measurements from one of the distributions above.
    \item Compute the minimum amount of \emph{white noise} that one needs to add to each element of the $g$-tuple in order to render them compatible. 
\end{enumerate}

The only exception to the recipe above is the model of induced random POVMs discussed in \cref{sec:random-POVMs}, where noise is added by tuning the ancilla dimension $n$. 

\subsubsection*{Random dichotomic projective measurements}
A central object in our study is the \emph{Haar-distributed random projection}. Let $d$ be a positive integer, and let $r \leq d$. A rank-$r$ projection $P$ on $\C^d$ is a Hermitian operator satisfying $P^2 = P$, $P = P^*$, and $\operatorname{Tr}(P) = r$. To sample a random projection of rank $r$, we proceed as follows: let $U$ be a random unitary matrix drawn from the Haar measure on the unitary group $U(d)$, and let $P_0$ be the fixed projection onto the first $r$ standard basis vectors, i.e., $P_0 = \sum_{i=1}^r \ketbra{i}{i}$. Then, the random projection is defined as
\[
P = U P_0 U^*.
\]
The distribution of $P$ is invariant under conjugation by unitaries, and is called the \emph{Haar measure} on the set of rank-$r$ projections.

Given such a projection $P$, we can define a \emph{random dichotomic projective measurement} as the pair $(P, I - P)$, which forms a dichotomic measurement with outcomes corresponding to the subspaces onto which $P$ and $I-P$ project. In this work, we are interested in the (in-)compatibility properties of tuples of independent random dichotomic projective measurements, i.e. when each projection is defined by an independent Haar-random unitary.

\subsubsection*{Random basis measurements}
Another important model of random projective measurement is the \emph{random basis measurement}. In this model, we consider a measurement in a randomly chosen orthonormal basis of $\C^d$. More precisely, let $U$ be a random unitary matrix drawn from the Haar measure on the unitary group $U(d)$. The columns of $U$, denoted by $\{ \psi_i \}_{i\in[d]}$, form an orthonormal basis of $\C^d$. The associated projective measurement is then given by the collection of rank-$1$ projections
\[
M_i = \ketbra{\psi_i}{\psi_i} \qquad \forall i\in[d].
\]
The measurement $(M_i)_{i\in[d]}$ is called a \emph{random basis measurement}. Sampling such a measurement amounts to sampling a Haar-random unitary $U$ and taking the projections onto its columns.

Random basis measurements are another natural form of projective measurements to consider, different from the random dichotomic projective measurements discussed above. They play a central role in quantum information theory, for example in the study of mutually unbiased bases and quantum state tomography. In this work, we are interested in the (in-)compatibility properties of tuples of independent random basis measurements, i.e. when each basis is defined by an independent Haar-random unitary.

\subsubsection*{Random induced measurements}
A final fundamental model of random (non projective) measurement is given by the \emph{induced random POVMs}, as introduced in~\cite{heinosaari2020random}. This construction is motivated by the physical scenario where a quantum system interacts with an environment (ancilla), and a projective measurement is performed on the joint system. The effective measurement on the system alone is then described by a POVM whose elements are obtained by tracing out the ancilla degrees of freedom.

More precisely, let $\mathcal{H}_S = \C^d$ be the Hilbert space of the system, and $\mathcal{H}_A = \C^n$ the Hilbert space of the ancilla (environment). Consider a projective measurement $(Q_i)_{i\in[k]}$ on the joint space $\mathcal{H}_S \otimes \mathcal{H}_A$, where each $Q_i$ is a rank-$1$ projection. To obtain a random POVM on the system, we proceed as follows:
\begin{enumerate}
    \item Sample a random orthonormal family $\{ \psi_i \}_{i\in[k]}$ in $\mathcal{H}_S \otimes \mathcal{H}_A$, for example as the first $k$ columns of a Haar-random unitary on $\C^{dn}$.
    \item For each $i\in[k]$, define the POVM element
    \[
    M_i = \Tr_A \left( \ketbra{\psi_i}{\psi_i} \right),
    \]
    where $\Tr_A:=\mathrm{id}_S\otimes\Tr_A$ denotes the partial trace over the ancilla.
\end{enumerate}
The resulting collection $(M_i)_{i\in[k]}$ forms a POVM on $\mathcal{H}_S$, i.e., $M_i \geq 0$ for each $i\in[k]$ and $\sum_{i=1}^k M_i = I_d$. The distribution of such POVMs, induced by the Haar measure on the joint system, is called the \emph{induced measure} on the set of POVMs with $k$ outcomes and ancilla dimension $n$.

The parameters of this model are the number $k$ of outcomes and the dimension $n$ of the ancilla. By varying $n$, one can interpolate between different regimes: for $n=1$, the model reduces to a random projective measurement, while for large $n$, the POVM elements become increasingly mixed, corresponding to the addition of more noise to the measurement.

\subsubsection*{Proof strategy}
In the remainder of this paper, we prove two main types of (in-)compatibility results concerning the three models of random POVMs described above:
\begin{enumerate}
    \item asymptotic ones, which hold with probability $1$ in the limit where the underlying dimension $d$ is infinite (cf.~Sections \ref{sec:two-projections}, \ref{sec:more-proj-infinite}, \ref{sec:random-bases-two} and \ref{sec:random-POVMs}),
    \item non-asymptotic ones, which hold with probability close to $1$ when the underlying dimension $d$ is finite but large (cf.~Sections \ref{sec:more-proj-finite} and \ref{sec:random-bases-more}).
\end{enumerate}
For the first kind of statements, we use tools from \emph{free probability}, which allow us to describe the almost sure infinite-dimensional limit of random matrices. As for the second one, we need to quantify the probability that a large but finite-size random matrix deviates from its limiting behavior, which is done thanks to \emph{concentration of measure} techniques.

\section{Compatibility and incompatibility of two random dichotomic projective measurements} \label{sec:two-projections}

In this section we focus on the (in-)compatibility of $2$ random dichotomic projective measurements. The case of an arbitrary number $g$ of such measurements is treated in \cref{sec:more-projections}.

\subsection{General facts: compressions} \label{sec:compression}

For $g$ dichotomic observables $A_1, \ldots, A_g$ on $\mathbb C^d$, we can define the following convex optimization problem, which is in fact an SDP (see also the end of Section \ref{sec:incompatibility-witnesses}):
\begin{align*}
    \lambda(A) = \text{maximize} \quad & \sum_{x=1}^g \Tr\left(A_xY_x\right) \\
    \text{subject to} \quad & \sum_{x=1}^g \epsilon_x Y_x \leq \rho \qquad \forall \epsilon \in \{\pm 1\}^g\\
    \quad & \rho\in M_d^{\mathrm{sa}}(\mathbb C) \\
    \quad & \rho \geq 0 \\
    \quad & \!\Tr(\rho) = 1 \\
    \quad & Y_x \in M_d^{\mathrm{sa}}(\mathbb C) \qquad \forall x \in [g].
\end{align*}
We know from \cite[Eq.~(7)]{bluhm2022tensor} that $\lambda(A)$ is equal to the so-called \emph{compatibility norm} of $A$, defined in \cite[Eq.~(3)]{bluhm2022tensor}. It is also shown there that $1/\lambda(A)=\tilde\tau(A)$, where
$$\tilde\tau(A) := \max\left\{ t\geq 0 \, : \, \left(E^{(t)}_{\pm|x} := \frac{I \pm t A_x}{2}\right)_{x\in[g]} \text{ are compatible}\right\} .$$
We see from the above definition that, if $0\leq\tilde\tau(A)\leq 1$, then $\tilde\tau(A)$ coincides with $\tau(A)$, the compatibility degree of the POVMs
\begin{equation*}
    \left(E_{\pm|x} := \frac{I \pm A_x}{2}\right)_{x\in[g]},
\end{equation*}
as introduced in Definition \ref{def:incmomp-degree} (where we have written $\tau(A)$ instead of $\tau(E)$ in a slight abuse of notation). But $\tilde\tau(A)$ may take values strictly greater than $1$ for `very compatible' dichotomic observables, contrary to $\tau(A)$ which is equal to $1$ for all compatible dichotomic observables. 

\begin{proposition} \label{prop:compatibility-compression}
    For any $k\in\mathbb N$ and any isometry $V: \C^k \to \C^d$, we have that, for any $g$-tuple of dichotomic observables $A$ on $\mathbb C^d$,
    $$\tau(A) \leq \tau(V^*AV).$$
\end{proposition}
In other words, Proposition \ref{prop:compatibility-compression} tells us that a $g$-tuple of observables cannot be more compatible than any of their compressions by an isometry.

\begin{proof}
    Let $t = \tau(A)$. Since the $E^{(t)}_{\pm|x}$ are compatible, there exists a joint POVM $(C_\epsilon)_{\epsilon\in\{\pm 1\}^g}$ such that 
    $$\frac{I_d+t A_x}{2} = \sum_{\epsilon \, : \, \epsilon_x = +1} C_\epsilon \qquad \forall x \in [g].$$
    Hence, we also have 
    $$\frac{I_k+t V^*A_xV}{2} = \sum_{\epsilon \, : \, \epsilon_x = +1} V^*C_\epsilon V \qquad \forall x \in [g].$$
    Since $(V^* C_\epsilon V)_{\epsilon\in\{\pm 1\}^g}$ is a compressed joint POVM, this proves that $\tau(V^*AV) \geq t$.
\end{proof}

\begin{proposition} \label{prop:compatibility-Lipschitz}
    The functional $\lambda$ is Lipschitz with respect to the $\ell_1^g(S_\infty^d)$ norm, i.e., for any $g$-tuples of dichotomic observables $A,B$ on $\mathbb C^d$,
    $$\left|\lambda(A) - \lambda(B)\right| \leq \sum_{x=1}^g \|A_x -B_x\|_\infty.$$
\end{proposition}
\begin{proof}
    In the SDP introduced at the beginning of this section, let $\lambda(A)$ be attained for operators $Y$ and $\rho$. Using the same feasible points for the SDP corresponding to $\lambda(B)$, we have
    $$\lambda(A) - \lambda(B) \leq \sum_{x=1}^g \Tr\left((A_x - B_x)Y_x\right) \leq \sum_{x=1}^g \|A_x -B_x\|_\infty \|Y_x\|_1.$$
    We claim that, for all $x \in [g]$, $\|Y_x\|_1 \leq 1$. Indeed, we have 
    $$Y_x = \frac{1}{2^{g-1}} \sum_{\epsilon \, : \, \epsilon_x = +1} \sum_{y=1}^g \epsilon_y Y_y,$$
    and moreover
    $$-\rho \leq \sum_{y=1}^g \epsilon_y Y_y \leq \rho.$$
    The latter inequality proves that $\|\sum_{y=1}^g \epsilon_y Y_y\|_1 \leq \Tr(\rho) = 1$, hence, taking averages, $\|Y_y\|_1 \leq 1$ for all $y\in[g]$. The reversed inequality follows in a similar way.
\end{proof}

Recall from \cite[Section VIII.B]{bluhm2018joint} that, for the generalized Pauli dichotomic observables $\sigma=(\sigma_x)_{x \in [g]}$, we have 
$$\tau(\sigma) = \frac{1}{\sqrt g}.$$
Hence, we have the following consequence of Propositions \ref{prop:compatibility-compression} and \ref{prop:compatibility-Lipschitz}.

\begin{corollary} \label{cor:compatibility-approximate}
    Let $A$ be a $g$-tuple of dichotomic observables for which there exist an isometry $V:\C^k \to \C^d$ and $0<\epsilon<\sqrt{g}$ such that 
    $$\sum_{x=1}^g \|V^* A_x V - \sigma_x\|_\infty \leq \epsilon,$$
    where $k$ is the dimension of the $g$-tuple of generalized Pauli dichotomic observables. Then 
    $$\tau(A) \leq \frac 1 {\sqrt g-\epsilon}.$$
\end{corollary}
\begin{proof}
    Since we have $\tilde\tau(\sigma)=\tau(\sigma) = 1/\sqrt g$, using the Lipschitz property from \cref{prop:compatibility-Lipschitz} we get that, for $0<\epsilon<\sqrt{g}$,
    $$\left|\frac{1}{\tilde\tau(V^*AV)} - \sqrt{g}\right| \leq \epsilon \implies \frac{1}{\tilde\tau(V^*AV)} \geq \sqrt{g} - \epsilon\implies \tilde\tau(V^*AV) \leq \frac 1 {\sqrt g-\epsilon}.$$
    Observing that we always have $\tau\leq\tilde\tau$, the latter inequality implies that
    $$\tau(V^*AV) \leq \frac 1 {\sqrt g-\epsilon}.$$
    And the conclusion follows from \cref{prop:compatibility-compression}.
\end{proof}

We can apply this in the setting of $g=2$ unbiased random observables in the following way. Suppose for simplicity that $d$ is even and define, for $x=1,2$, 
$$A_x = U_x \operatorname{diag}(\underbrace{1, \ldots, 1}_{d/2 \text{ times}}, \underbrace{-1, \ldots, -1}_{d/2 \text{ times}}) U_x^*,$$
where $U_1,U_2$ are independent random Haar-distributed unitary matrices. Without loss of generality, we can assume that $U_1 = I_d$ and $U_2 = U$, where $U$ is a random Haar-distributed unitary matrix. In order to prove that $A_1$ and $A_2$ are close to maximally incompatible, we would need to find an isometry $V:\C^2 \to \C^d$ such that 
$$V^*A_1 V \approx \begin{bmatrix} 1 & 0 \\ 0 & -1\end{bmatrix} \quad \text{ and } \quad V^*A_2 V \approx \begin{bmatrix} 0 & 1 \\ 1 & 0\end{bmatrix}.$$
This is what we proceed to do in the next subsection, actually covering the case of random biased observables as well.
 
\subsection{Incompatibility of two random dichotomic projective measurements} 
\label{sec:two-proj-incomp}

The following lemma is a re-writing of a standard result on the incompatibility of unbiased dichotomic measurements on $\C^2$, which can for instance be found in \cite[Section 3]{Heinosaari_2020}.

\begin{lemma} \label{lem:incompatibility-Pauli}
Let $M=m_XX+m_YY+m_ZZ$ and $N=n_XX+n_YY+n_ZZ$ be two unbiased observables on $\C^2$, with $m,n\in\mathbb R^3$ such that $\|m\|,\|n\|\leq 1$. Then,
\[ \tau(M,N) = \frac{2}{\|m+n\|+\|m-n\|} . \]
\end{lemma}

Having in mind this fact about incompatibility of Pauli observables, we our now ready to establish the main results of this subsection, namely upper bounds on the compatibility degree of two independent random dichotomic PVMs, which will be obtained by compressing their corresponding observables to Pauli matrices. In what follows, we will be considering $E,F$ two independent random subspaces of $\mathbb C^d$ of respective dimensions $\lfloor\alpha d\rfloor,\lfloor\beta d\rfloor$ for some parameters $0<\alpha,\beta<1$. In order to make the presentation lighter, we will drop floor functions in the proofs, and write things as if $\alpha d,\beta d$ were integers. This slight abuse has no impact on the results since, as we will see, the only quantities that matter are the asymptotic ratios $\dim(E)/d,\dim(F)/d$ as $d\to\infty$.

Before we state our results, let us clarify one last point. When we say that a sequence of $2\times 2$ random matrices $(M_d)_{d\in\mathbb N}$ converges in probability towards a given $2\times 2$ matrix $M_*$, we mean that, for any $\epsilon>0$,
\[ \P\left(\left\|M_d-M_*\right\|_\infty\leq\epsilon\right) \xrightarrow[d \to \infty]{} 1. \]

We start with the case where the dimensions $(\lfloor\alpha/d\rfloor,\lfloor\beta/d\rfloor)$ of the two subspaces are close to $(d/2,d/2)$ (in the sense that $(\alpha,\beta)$ is inside the disc of radius $1/2$ centered in $(1/2,1/2)$). We show that, under this assumption, we can find a random isometry $V:\C^2\to\C^d$ that sends $A$ on $Z$ and $B$ with high probability close to $X$.
\begin{theorem} \label{th:two-proj-isometry}
    Fix $0<\alpha,\beta<1$ and suppose that
    \[ \left(\alpha-\frac{1}{2}\right)^2+\left(\beta-\frac{1}{2}\right)^2 \leq \frac{1}{4}. \]
    Let $E,F\subset\C^d$ be independent uniformly distributed subspaces of dimensions $\lfloor\alpha d\rfloor,\lfloor\beta d\rfloor$, respectively, and set $A=P_E-P_E^\perp$ and $B=P_F-P_F^\perp$. Then, there exists an isometry $V:\C^2\to\C^d$ such that
    \[ V^*AV=Z \quad \text{and} \quad V^*BV \xrightarrow[d \to \infty]{} X \text{ in probability}. \]
\end{theorem}

\begin{proof}
    Our goal is to find $V:\C^2\to\C^d$ such that $V^*V=\ketbra{1}{1}+\ketbra{2}{2}$, $V^*AV=\ketbra{1}{1}-\ketbra{2}{2}$ and $V^*BV\to\ketbra{1}{2}+\ketbra{2}{1}$ as $d\to\infty$. It is easy to see that, for any unit vectors $v\in E$ and $v'\in E^\perp$, if we take $V=\ketbra{v}{1}+\ketbra{v'}{2}$, then it automatically satisfies the first two conditions. So what remains to be understood is whether these unit vectors $v,v'$ can be picked in such a way that the third condition is also satisfied. Observe that we can rewrite 
\begin{align*} V^*BV & = \bra{v}P_F-P_F^\perp\ket{v}\ketbra{1}{1} + \bra{v'}P_F-P_F^\perp\ket{v'}\ketbra{2}{2} + \bra{v}P_F-P_F^\perp\ket{v'}\ketbra{1}{2} + \bra{v'}P_F-P_F^\perp\ket{v}\ketbra{2}{1} \\
& = \left(2\bra{v}P_F\ket{v}-1\right)\ketbra{1}{1} + \left(2\bra{v'}P_F\ket{v'} -1\right)\ketbra{2}{2} + 2\bra{v}P_F\ket{v'}\ketbra{1}{2} + 2\bra{v'}P_F\ket{v}\ketbra{2}{1} . \end{align*}
So in order to have $V^*BV\approx\ketbra{1}{2}+\ketbra{2}{1}$, we need 
\[ \bra{v}P_F\ket{v} \approx \bra{v'}P_F\ket{v'} \approx \bra{v}P_F\ket{v'} \approx \bra{v'}P_F\ket{v} \approx \frac{1}{2}. \]

Now, setting $\gamma=\min(\alpha,\beta, 1-\alpha,1-\beta)$, we know from e.g.~\cite[Proposition 1.1 and Lemma 1.3]{aubrun2021} that we can always find orthonormal bases $\{e_1,\ldots,e_{\alpha d}\}$ of $E$ and $\{f_1,\ldots,f_{\beta d}\}$ of $F$ as well as ordered angles $0\leq\theta_1\leq\ldots\leq\theta_{\gamma d}\leq 0$ such that, for each $i,j\in[\gamma d]$, $\braket{e_i}{f_j}=\delta_{i,j}\cos(\theta_i)$. The $\theta_i$'s are called the (non-zero) principal angles between the subspaces $E$ and $F$. We then have that we can also find an orthonormal basis $\{e'_1,\ldots,e'_{(1-\alpha)d}\}$ of $E^\perp$ such that, for each $i,j\in[\gamma d]$, $\braket{e'_i}{f_j}=\delta_{i,j}\sin(\theta_i)$. Hence, if we choose $v=e_i$ and $v'=e'_i$ for some $i\in[\gamma d]$, we get
\[ \bra{v}P_F\ket{v}=\cos^2(\theta_i) \quad \bra{v'}P_F\ket{v'}=\sin^2(\theta_i) \quad \bra{v}P_F\ket{v'}=\bra{v'}P_F\ket{v}=\cos(\theta_i)\sin(\theta_i). \]
So we are done if we can argue that there exists $i\in[\gamma d]$ such that $\theta_i\approx\pi/4$. 

If we define the empirical distribution of the $\theta_i$'s as 
\[ \mu_\theta^{(d)}=\frac{1}{d}\sum_{i=1}^{\gamma d}\delta_{\theta_i}, \]
we know from \cite[Theorem 3.1]{aubrun2021} that, as $d\to\infty$, $\mu_\theta^{(d)}$ converges weakly in probability to a distribution supported on $[\arccos\sqrt{\smash[b]{\lambda_{\alpha,\beta}^+}},\arccos\sqrt{\smash[b]{\lambda_{\alpha,\beta}^-}}]$, where 
\[ \lambda_{\alpha,\beta}^{\pm} = \alpha+\beta-2\alpha\beta \pm 2\sqrt{\alpha(1-\alpha)\beta(1-\beta)} . \]
Since it can easily be checked that the condition $(\alpha-1/2)^2+(\beta-1/2)^2\leq 1/4$ is equivalent to the conditions $\lambda_{\alpha,\beta}^+\geq 1/2$ and  $\lambda_{\alpha,\beta}^-\leq 1/2$, we have that, under this assumption, $\pi/4\in[\arccos\sqrt{\smash[b]{\lambda_{\alpha,\beta}^+}},\arccos\sqrt{\smash[b]{\lambda_{\alpha,\beta}^-}}]$. Hence, for any $\epsilon>0$, the probability that there exists $i_d\in[\gamma d]$ such that $\theta_{i_d}\in[\pi/4-\epsilon,\pi/4+\epsilon]$ goes to $1$ as $d\to\infty$. Choosing such $i_d$ in the definition of $V$ we get that the probability that $\|V^*BV-X\|_\infty\leq C\epsilon$ goes to $1$ as $d\to\infty$, for $C<\infty$ some absolute constant. This is exactly the announced result (up to relabeling $C\epsilon$ into $\epsilon$).
\end{proof}

\begin{corollary} \label{cor:two-proj-max-incomp}
    Fix $0<\alpha,\beta<1$ and suppose that
    \[ \left(\alpha-\frac{1}{2}\right)^2+\left(\beta-\frac{1}{2}\right)^2 \leq \frac{1}{4}. \]
    Let $E,F\subset\C^d$ be independent uniformly distributed subspaces of dimensions $\lfloor\alpha d\rfloor,\lfloor\beta d\rfloor$, respectively, and set $A=P_E-P_E^\perp$ and $B=P_F-P_F^\perp$. Then, for any $0<\epsilon\leq 1/\sqrt{2}$,
    \[ \P\left( \tau(A,B) \leq \frac{1}{\sqrt{2}} + \epsilon \right) \xrightarrow[d \to \infty]{} 1.  \]
\end{corollary}

\begin{proof}
    By Proposition \ref{prop:compatibility-compression}, we know that, for any isometry $V:\C^2\to\C^d$,
    \[ \tau(A,B) \leq \tau(V^*AV,V^*BV). \]
    Now, we also have by Theorem \ref{th:two-proj-isometry} that there exists an isometry $V:\C^2\to\C^d$ such that $V^*AV=Z$ and, with probability going to $1$ as $d\to\infty$, $\|V^*BV-X\|_\infty\leq\epsilon$. Hence, by Corollary \ref{cor:compatibility-approximate}, with probability going to $1$ as $d\to\infty$,
    \[ \tau(A,B) \leq \frac{1}{1/\tau(Z,X) - \epsilon} = \frac{1}{\sqrt{2} - \epsilon}, \]
    where the last equality is by Lemma \ref{lem:incompatibility-Pauli}, with $m=(0,0,1)$ and $n=(1,0,0)$, so that $\|m+n\|=\|m-n\|=\sqrt{2}$. It is then easy to check that $1/(\sqrt{2}-\epsilon)\leq 1/\sqrt{2}+\epsilon$ if $\epsilon\leq 1/\sqrt{2}$, which concludes the proof. 
\end{proof}

We now move on to the case where the dimensions $(\lfloor\alpha/d\rfloor,\lfloor\beta/d\rfloor)$ of the two subspaces are far from $(d/2,d/2)$ (in the sense that $(\alpha,\beta)$ is outside the disc of radius $1/2$ centered in $(1/2,1/2)$), and actually focus only on the case where $\beta=\alpha$ for simplicity. We show that, under this assumption, we can find a random isometry $V:\C^2\to\C^d$ that sends $A$ on $Z$ and $B$ with high probability close to a specific combination of $Z$ and $X$.
\begin{theorem} \label{th:two-proj-isometry'}
    Fix $0<\alpha<1$ and suppose that
    \[ \alpha > \frac{1}{2}\left(1+\frac{1}{\sqrt{2}}\right) \quad \text{or} \quad \alpha < \frac{1}{2}\left(1-\frac{1}{\sqrt{2}}\right). \]
    Let $E,F\subset\C^d$ be independent uniformly distributed subspaces of dimension $\lfloor\alpha d\rfloor$, and set $A=P_E-P_E^\perp$ and $B=P_F-P_F^\perp$. Then, there exists an isometry $V:\C^2\to\C^d$ such that
    \[ V^*AV=Z \quad \text{and} \quad V^*BV \xrightarrow[d \to \infty]{}(2\lambda_\alpha-1)Z+2\sqrt{\lambda_\alpha(1-\lambda_\alpha)}X \text{ in probability}, \text{ where } \lambda_\alpha=4\alpha(1-\alpha). \] 
\end{theorem}

\begin{proof}
    We follow the exact same reasoning as in the proof of Theorem \ref{th:two-proj-isometry}, and we use the same notations as those introduced there. Namely, we define the isometry $V:\C^2\to\C^d$ as $V=\ketbra{e_i}{1}+\ketbra{e_i'}{2}$ for some $i\in[\gamma d]$, where $\gamma=\min(\alpha,1-\alpha)$. That way, we have $V^*AV=Z$ and
    \[ V^*BV = \left(2\cos^2(\theta_i)-1\right)Z+2\cos(\theta_i)\sin(\theta_k)X. \]
    Now, we know from \cite[Theorem 3.1]{aubrun2021} that, as $d\to\infty$, $\mu_\theta^{(d)}$ converges weakly in probability to a distribution supported on $[\arccos\sqrt{\smash[b]{\lambda_\alpha}},\pi/2]$, where $\lambda_\alpha=4\alpha(1-\alpha)$. Hence, for any $\epsilon>0$, the probability that there exists $i_d\in[\gamma d]$ such that $\theta_{i_d}\in[\arccos\sqrt{\smash[b]{\lambda_\alpha}}-\epsilon,\arccos\sqrt{\smash[b]{\lambda_\alpha}}+\epsilon]$ goes to $1$ as $d\to\infty$. Choosing such $i_d$ in the definition of $V$ we get that the probability that $\|V^*BV-(2\lambda_\alpha-1)Z-2\sqrt{\smash[b]{\lambda_\alpha(1-\lambda_\alpha)}}X\|_\infty\leq C\epsilon$ goes to $1$ as $d\to\infty$, for $C<\infty$ some absolute constant. This is exactly the announced result (up to relabeling $C\epsilon$ into $\epsilon$). 
\end{proof}

\begin{corollary} \label{cor:upper-two-proj-unbalanced}
    Fix $0<\alpha<1$ and suppose that
    \[ \alpha > \frac{1}{2}\left(1+\frac{1}{\sqrt{2}}\right) \quad \text{or} \quad \alpha < \frac{1}{2}\left(1-\frac{1}{\sqrt{2}}\right). \]
    Let $E,F\subset\C^d$ be independent uniformly distributed subspaces of dimension $\lfloor\alpha d\rfloor$, and set $A=P_E-P_E^\perp$ and $B=P_F-P_F^\perp$. Then, for any $0<\epsilon\leq 2\sqrt{\lambda_\alpha}\sqrt{1-\lambda_\alpha}/(\sqrt{\lambda_\alpha}+\sqrt{1-\lambda_\alpha})$,
    \[ \P\left( \tau(A,B) \leq \frac{1}{\sqrt{\lambda_\alpha}+\sqrt{1-\lambda_\alpha}}+\epsilon \right) \xrightarrow[d \to \infty]{} 1, \text{ where } \lambda_\alpha=4\alpha(1-\alpha). \]
\end{corollary}

\begin{proof}
    By Proposition \ref{prop:compatibility-compression}, we know that, for any isometry $V:\C^2\to\C^d$,
    \[ \tau(A,B) \leq \tau(V^*AV,V^*BV). \]
    Now, we also have by Theorem \ref{th:two-proj-isometry'} that there exists an isometry $V:\C^2\to\C^d$ such that $V^*AV=Z$ and, with probability going to $1$ as $d\to\infty$, $\|V^*BV-(2\lambda_\alpha-1)Z-2\sqrt{\smash[b]{\lambda_\alpha(1-\lambda_\alpha)}}X\|_\infty\leq\epsilon$. Hence, by Corollary \ref{cor:compatibility-approximate}, with probability going to $1$ as $d\to\infty$,
    \[ \tau(A,B) \leq \frac{1}{1/\tau(Z,(2\lambda_\alpha-1)Z+2\sqrt{\lambda_\alpha(1-\lambda_\alpha)}X) - \epsilon} = \frac{1}{\sqrt{\lambda_\alpha}+\sqrt{1-\lambda_\alpha} - \epsilon}, \]
    where the last equality is by Lemma \ref{lem:incompatibility-Pauli}, with $m=(0,0,1)$ and $n=(2\sqrt{\smash[b]{\lambda_\alpha(1-\lambda_\alpha)}},0,2\lambda_\alpha-1)$, so that $\|m+n\|=2\sqrt{\smash[b]{\lambda_\alpha}}$ and $\|m-n\|=2\sqrt{\smash[b]{1-\lambda_\alpha}}$. It is then easy to check that $1/(\sqrt{\lambda_\alpha}+\sqrt{1-\lambda_\alpha}-\epsilon)\leq 1/(\sqrt{\lambda_\alpha}+\sqrt{1-\lambda_\alpha})+\epsilon$ if $\epsilon\leq 2\sqrt{\lambda_\alpha}\sqrt{1-\lambda_\alpha}/(\sqrt{\lambda_\alpha}+\sqrt{1-\lambda_\alpha})$, which concludes the proof.
\end{proof}

\begin{remark}
    All the results in this subsection are stated for $E,F\subset\C^d$ being independent uniformly distributed subspaces of dimension $\lfloor\alpha d\rfloor,\lfloor\beta d\rfloor$. They would remain true, slightly more generally, if we know that $\dim(E)/d\to\alpha$ and $\dim(F)/d\to\beta$ as $d\to\infty$.
\end{remark}

\subsection{Compatibility of two random dichotomic projective measurements} 
\label{sec:two-proj-comp}

We start by recalling the so-called Jordan product criterion for compatibility of two measurements, established in \cite{heinosaari2013simple}.

\begin{lemma} \label{lem:Jordan}
    Let $\mathrm M=(M_i)_{i\in[k]}$ and $\mathrm N=(N_j)_{j\in[l]}$ be two POVMs on $\C^d$. If, for all $i\in[k]$ and $j\in[l]$,
    \[ M_iN_j+N_jM_i \geq 0, \]
    then $\mathrm M$ and $\mathrm N$ are compatible.
\end{lemma}

We now apply Lemma \ref{lem:Jordan} to two noisy dichotomic PVMs 
\[ \mathrm P_E^{(t)}=\left(tP_E+(1-t)\frac{I}{2},tP_E^\perp+(1-t)\frac{I}{2}\right) \quad \text{and} \quad \mathrm P_F^{(t)}=\left(tP_F+(1-t)\frac{I}{2},tP_F^\perp+(1-t)\frac{I}{2}\right), \] 
where $P_E,P_F$ are projectors on $\C^d$ and $0\leq t\leq 1$. We get that, if
\begin{equation} \label{eq:Jordan} t^2\left(P_E'P_F'+P_F'P_E'\right) + t(1-t)\left(P_E'+P_F'\right) + (1-t)^2\frac{I}{2} \geq 0 \end{equation}
for all $P_E'\in\{P_E,P_E^\perp\}$ and $P_F'\in\{P_F,P_F^\perp\}$, then $\mathrm P_E^{(t)}$ and $\mathrm P_F^{(t)}$ are compatible.

With this compatibility criterion at hand, we are now ready to establish the main results of this subsection, namely lower bounds on the compatibility degree of two independent random dichotomic PVMs, which will turn out to match the upper bounds obtained in the previous subsection. As explained there, we will once again make the slight abuse of denoting the dimensions of the considered subspaces $\alpha d,\beta d$ instead of $\lfloor\alpha d\rfloor,\lfloor\beta d\rfloor$. We will analogously split our analysis into two distinct cases, namely: $(\alpha,\beta)$ either inside or outside of the disc of radius $1/2$ centered in $(1/2,1/2)$ (and in the latter case we will only focus on the case where $\beta=\alpha$ for simplicity). Note also that, in the former case, Corollary \ref{cor:two-proj-max-incomp} provides an asymptotic upper bound which is equal to the minimal compatibility degree of two dichotomic POVMs (namely $1/\sqrt{2}$), so we already know that it is actually tight. The point of the following analysis is thus only to show that the Jordan product criterion performs optimally in this case as it does provide a tight asymptotic lower bound.

\begin{lemma} \label{lem:spectrum}
    Fix $0<\alpha,\beta<1$ and suppose that
    \[ \left(\alpha-\frac{1}{2}\right)^2+\left(\beta-\frac{1}{2}\right)^2 \leq \frac{1}{4}. \]
    Let $E,F\subset\C^d$ be independent uniformly distributed subspaces of dimensions $\lfloor\alpha d\rfloor,\lfloor\beta d\rfloor$, respectively. Then, for any $1/(2\sqrt{\smash[b]{\lambda_{\alpha,\beta}^+}})\leq t\leq 1/(2\sqrt{\smash[b]{\lambda_{\alpha,\beta}^-}})$ and for $P_E'\in\{P_E,P_E^\perp\}$, $P_F'\in\{P_F,P_F^\perp\}$, we have, for all $\epsilon>0$,
    \[ \P\left( \lambda_{\min}\left(t^2\left(P_E'P_F'+P_F'P_E'\right) + t(1-t)\left(P_E'+P_F'\right) + (1-t)^2\frac{I}{2} \right) \geq \frac{1}{4}-\frac{t^2}{2}-\epsilon \right) \xrightarrow[d \to \infty]{} 1 .  \]
\end{lemma}

\begin{proof}
    We focus on the case where $P_E'=P_E$ and $P_F'=P_F$, as the three other cases are completely analogous. We set $\gamma=\min(\alpha,\beta,1-\alpha,1-\beta)$ and we denote by $0\leq\theta_1\leq\cdots\leq\theta_{\gamma d}\leq\pi/2$ the (non-zero) principal angles between $E$ and $F$. Clearly, $P_E+P_F$ and $I$ do not have negative eigenvalues. As for $P_EP_F+P_FP_E$, we know from \cite[Proposition 1.2]{aubrun2021} that its only negative eigenvalues are $-\cos(\theta_i)(1-\cos(\theta_i))$, attained on $\mathrm{span}(e_i,f_i)$, for $i\in[\gamma d]$. Additionally, for each $i\in[\gamma d]$, the eigenvalues of $P_E+P_F$ restricted to the subspace $\mathrm{span}(e_i,f_i)$, are $1\pm\cos(\theta_i)$, hence lower bounded by $1-\cos(\theta_i)$. We thus have
    \begin{align*} & \lambda_{\min}\left(t^2\left(P_EP_F+P_FP_E\right) + t(1-t)\left(P_E+P_F\right) + (1-t)^2\frac{I}{2} \right) \\
    & \quad \geq \min\left\{ -t^2\cos(\theta_i)(1-\cos(\theta_i)) + t(1-t)(1-\cos(\theta_i)) + \frac{(1-t)^2}{2} : i\in[\gamma d] \right\} . \end{align*}
    Next, we know from \cite[Theorem 3.1]{aubrun2021} that, as $d\to\infty$, $\mu_\theta^{(d)}$ converges weakly in probability to a distribution supported on $[\arccos\sqrt{\smash[b]{\lambda_{\alpha,\beta}^+}},\arccos\sqrt{\smash[b]{\lambda_{\alpha,\beta}^-}}]$. What is more, we know from \cite[Theorem 1.4]{collins2014strong} that there actually is strong convergence, meaning that there are asymptotically no outliers outside of this support. Indeed, $\mu_\theta^{(d)}$ is the push-forward by the function $x\mapsto\arccos\sqrt{x}$ of the spectral distribution of $P_EP_FP_E$ restricted to its support (since the latter has eigenvalues $\cos^2(\theta_i)$, $i\in[\gamma d]$). Now, $P_EP_FP_E$ is distributed as $PUQU^*P$ for $P,Q$ fixed projectors and $U$ a Haar-distributed unitary independent from $P,Q$, so it is a polynomial in random matrices that satisfy the assumptions of \cite[Theorem 1.4]{collins2014strong}. Therefore, its spectral distribution converges strongly (towards that of the corresponding polynomial in free elements), which implies the same for its push-forward.
    
    Hence, setting
    \[ m_{\alpha,\beta}(t) = \min\left\{ -t^2x(1-x) + t(1-t)(1-x) + \frac{(1-t)^2}{2} : \sqrt{\lambda_{\alpha,\beta}^-}\leq x\leq\sqrt{\lambda_{\alpha,\beta}^+}\right\}, \]
    we have that, for any $\epsilon>0$, with probability going to $1$ as $d\to\infty$,
    \[ \lambda_{\min}\left(t^2\left(P_EP_F+P_FP_E\right) + t(1-t)\left(P_E+P_F\right) + (1-t)^2\frac{I}{2} \right) \geq m_{\alpha,\beta}(t)-\epsilon . \]
    Now, it is easy to check that $f_t:x\in[0,1]\mapsto -t^2x(1-x) + t(1-t)(1-x) + (1-t)^2/2$ attains its minimum on $[0,1]$ in $1/(2t)$. So if $\sqrt{\smash[b]{\lambda_{\alpha,\beta}^-}}\leq 1/(2t)\leq\sqrt{\smash[b]{\lambda_{\alpha,\beta}^+}}$, then $f_t$ attains its minimum on $[\sqrt{\smash[b]{\lambda_{\alpha,\beta}^-}},\sqrt{\smash[b]{\lambda_{\alpha,\beta}^+}}]$ in $1/(2t)$. By assumption on $t$, this is indeed satisfied, and we thus have 
    \[ m_{\alpha,\beta}(t) = f_t\left(\frac{1}{2t}\right)=\frac{1}{4}-\frac{t^2}{2}, \]
    which is exactly the announced result.
\end{proof}

\begin{corollary}
    Fix $0<\alpha,\beta<1$ and suppose that
    \[ \left(\alpha-\frac{1}{2}\right)^2+\left(\beta-\frac{1}{2}\right)^2 < \frac{1}{4}. \]
    Let $E,F\subset\C^d$ be independent uniformly distributed subspaces of dimensions $\lfloor\alpha d\rfloor,\lfloor\beta d\rfloor$, respectively, and set $A=P_E-P_E^\perp$ and $B=P_F-P_F^\perp$. Then, 
    \[ \P\left( \tau(A,B) \geq \frac{1}{\sqrt{2}} \right) \xrightarrow[d \to \infty]{} 1. \]
\end{corollary}

\begin{proof}
    By the Jordan product criterion, summarized in Eq.~\eqref{eq:Jordan} for dichotomic projective measurements, we have that, given $0\leq t\leq 1$, if 
    \begin{equation} \label{eq:Jordan-spectrum} \lambda_{\min}\left(t^2\left(P_E'P_F'+P_F'P_E'\right) + t(1-t)\left(P_E'+P_F'\right) + (1-t)^2\frac{I}{2} \right) \geq 0 \end{equation}
    for $P_E'\in\{P_E,P_E^\perp\}$ and $P_F'\in\{P_F,P_F^\perp\}$, then $\tau(A,B)\geq t$. By Lemma \ref{lem:spectrum}, we know that, for all $1/(2\sqrt{\smash[b]{\lambda_{\alpha,\beta}^+}})\leq t\leq1/(2\sqrt{\smash[b]{\lambda_{\alpha,\beta}^-}})$, for any $\epsilon>0$, with probability going to $1$ as $d\to\infty$,
    \[ \lambda_{\min}\left(t^2\left(P_E'P_F'+P_F'P_E'\right) + t(1-t)\left(P_E'+P_F'\right) + (1-t)^2\frac{I}{2} \right) \geq m_{\alpha,\beta}(t)-\epsilon, \]
    with $m_{\alpha,\beta}(t)=1/4-t^2/2$. Now, $m_{\alpha,\beta}(t)>0$ for all $1/(2\sqrt{\smash[b]{\lambda_{\alpha,\beta}^+}})\leq t< 1/\sqrt{2}$ (by assumption on $\alpha,\beta$, we indeed have $1/(2\sqrt{\smash[b]{\lambda_{\alpha,\beta}^+}})< 1/\sqrt{2}< 1/(2\sqrt{\smash[b]{\lambda_{\alpha,\beta}^-}})$). Hence, for all $1/(2\sqrt{\smash[b]{\lambda_{\alpha,\beta}^+}})\leq t<1/\sqrt{2}$, there exists $\epsilon>0$ such that $m_{\alpha,\beta}(t)-\epsilon\geq 0$, and thus, with probability going to $1$ as $d\to\infty$, the inequality in Eq.~\eqref{eq:Jordan-spectrum} is satisfied, which concludes the proof.
\end{proof}

\begin{lemma} \label{lem:spectrum'}
    Fix $0<\alpha<1$ and suppose that
    \[ \alpha > \frac{1}{2}\left(1+\frac{1}{\sqrt{2}}\right) \quad \text{or} \quad \alpha < \frac{1}{2}\left(1-\frac{1}{\sqrt{2}}\right). \]
    Let $E,F\subset\C^d$ be independent uniformly distributed subspaces of dimension $\lfloor\alpha d\rfloor$. Then, for any $0\leq t\leq 1/(\sqrt{\smash[b]{\lambda_\alpha}}+\sqrt{\smash[b]{1-\lambda_\alpha}})$ and for $P_E'\in\{P_E,P_E^\perp\}$, $P_F'\in\{P_F,P_F^\perp\}$, we have, for all $\epsilon>0$,
    \[ \P\left( \lambda_{\min}\left(t^2\left(P_E'P_F'+P_F'P_E'\right) + t(1-t)\left(P_E'+P_F'\right) + (1-t)^2\frac{I}{2} \right) \geq \left(\lambda_\alpha-\frac{1}{2}\right)t^2 - \sqrt{\lambda_\alpha}t+\frac{1}{2}-\epsilon \right) \xrightarrow[d \to \infty]{} 1. \]
\end{lemma}

\begin{proof}
We follow the exact same strategy as in the proof of Lemma \ref{lem:spectrum} and use some of the notations introduced there, focusing on the case where $P_E'=P_E$ and $P_F'=P_F$ and setting $\gamma=\min(\alpha,1-\alpha)$. As already explained, we have as a consequence of \cite[Proposition 1.2]{aubrun2021} that 
    \begin{align*} & \lambda_{\min}\left(t^2\left(P_EP_F+P_FP_E\right) + t(1-t)\left(P_E+P_F\right) + (1-t)^2\frac{I}{2} \right) \\
    & \quad \geq \min\left\{ -t^2\cos(\theta_i)(1-\cos(\theta_i)) + t(1-t)(1-\cos(\theta_i)) + \frac{(1-t)^2}{2} : i\in[\gamma d] \right\} . \end{align*}
    Next, we know from \cite[Theorem 3.1]{aubrun2021} that, as $d\to\infty$, $\mu_\theta^{(d)}$ converges weakly in probability to a distribution supported on $[\arccos\sqrt{\smash[b]{\lambda_{\alpha}}},\pi/2]$. What is more, we know from \cite[Theorem 1.4]{collins2014strong} that there actually is strong convergence, meaning that there are asymptotically no outliers outside of this support. Hence, setting
    \[ m_{\alpha}(t) = \min\left\{ -t^2x(1-x) + t(1-t)(1-x) + \frac{(1-t)^2}{2} : 0\leq x\leq\sqrt{\lambda_{\alpha}}\right\}, \]
    we have that, for any $\epsilon>0$, with probability going to $1$ as $d\to\infty$,
    \[ \lambda_{\min}\left(t^2\left(P_EP_F+P_FP_E\right) + t(1-t)\left(P_E+P_F\right) + (1-t)^2\frac{I}{2} \right) \geq m_{\alpha}(t)-\epsilon . \]
    Now, it is easy to check that $f_t:x\in[0,1]\mapsto -t^2x(1-x) + t(1-t)(1-x) + (1-t)^2/2$ is decreasing on $[0,1/(2t)]$. So if $1/(2t)\geq\sqrt{\smash[b]{\lambda_{\alpha}}}$, then $f_t$ attains its minimum on $[0,\sqrt{\smash[b]{\lambda_{\alpha}}}]$ in $\sqrt{\smash[b]{\lambda_{\alpha}}}$. By assumption on $t$, $1/(2t)\geq(\sqrt{\smash[b]{\lambda_\alpha}}+\sqrt{\smash[b]{1-\lambda_\alpha}})/2$, and by assumption on $\alpha$, $\sqrt{\smash[b]{1-\lambda_\alpha}}\geq\sqrt{\smash[b]{\lambda_\alpha}}$, so this is indeed satisfied, and we thus have 
    \[ m_{\alpha}(t) = f_t\left(\sqrt{\lambda_{\alpha}}\right)= \left(\lambda_\alpha-\frac{1}{2}\right)t^2 - \sqrt{\lambda_\alpha}t+\frac{1}{2}, \]
    which is exactly the announced result.
\end{proof}

\begin{corollary} \label{cor:lower-two-proj-unbalanced}
    Fix $0<\alpha<1$ and suppose that
    \[ \alpha > \frac{1}{2}\left(1+\frac{1}{\sqrt{2}}\right) \quad \text{or} \quad \alpha < \frac{1}{2}\left(1-\frac{1}{\sqrt{2}}\right). \]
    Let $E,F\subset\C^d$ be independent uniformly distributed subspaces of dimension $\lfloor\alpha d\rfloor$, and set $A=P_E-P_E^\perp$ and $B=P_F-P_F^\perp$. Then, 
    \[ \P\left( \tau(A,B) \geq \frac{1}{\sqrt{\lambda_\alpha}+\sqrt{1-\lambda_\alpha}} \right) \xrightarrow[d \to \infty]{} 1.  \]
\end{corollary}

\begin{proof}
    By the Jordan product criterion, summarized in Eq.~\eqref{eq:Jordan} for dichotomic projective measurements, we have that, given $0\leq t\leq 1$, if 
    \begin{equation} \label{eq:Jordan-spectrum'} \lambda_{\min}\left(t^2\left(P_E'P_F'+P_F'P_E'\right) + t(1-t)\left(P_E'+P_F'\right) + (1-t)^2\frac{I}{2} \right) \geq 0 \end{equation}
    for $P_E'\in\{P_E,P_E^\perp\}$ and $P_F'\in\{P_F,P_F^\perp\}$, then $\tau(A,B)\geq t$. By Lemma \ref{lem:spectrum'}, we know that, for all $0\leq t\leq 1/(\sqrt{\smash[b]{\lambda_\alpha}}+\sqrt{\smash[b]{1-\lambda_\alpha}})$, for any $\epsilon>0$, with probability going to $1$ as $d\to\infty$,
    \[ \lambda_{\min}\left(t^2\left(P_E'P_F'+P_F'P_E'\right) + t(1-t)\left(P_E'+P_F'\right) + (1-t)^2\frac{I}{2} \right) \geq m_{\alpha}(t)-\epsilon, \]
    with $m_\alpha(t)=(\lambda_\alpha-1/2)t^2 - \sqrt{\smash[b]{\lambda_\alpha}}t+1/2$. Now, $m_\alpha(t)>0$ for all $0\leq t< 1/(\sqrt{\smash[b]{\lambda_\alpha}}+\sqrt{\smash[b]{1-\lambda_\alpha}})$. Hence, for all $0\leq t<1/(\sqrt{\smash[b]{\lambda_\alpha}}+\sqrt{\smash[b]{1-\lambda_\alpha}})$, there exists $\epsilon>0$ such that $m_{\alpha}(t)-\epsilon\geq 0$, and thus, with probability going to $1$ as $d\to\infty$, the inequality in Eq.~\eqref{eq:Jordan-spectrum'} is satisfied, which concludes the proof. 
\end{proof}

\begin{remark}
    As already pointed out in the previous subsection, all the results in this subsection would remain true if we only impose that $\dim(E)/d\to\alpha$ and $\dim(F)/d\to\beta$ as $d\to\infty$, rather than $\dim(E)=\lfloor \alpha d\rfloor$ and $\dim(F)=\lfloor \beta d\rfloor$.
\end{remark}

\section{Incompatibility of random dichotomic projective measurements} \label{sec:more-projections}

In this section we study the incompatibility of an arbitrary number $g \geq 2$ of random dichotomic projective measurements. For this more general case, we cannot rely anymore on a careful analysis of the distribution of angles between the involved random subspaces, as we did in \cref{sec:two-projections}. We shall thus use a different (somehow rougher) technique, which we introduce next. 

\subsection{General facts: incompatibility witnesses} 

Let the dimension $d$ be even and let $P_1,\ldots,P_g$, be $g$ independent Haar-random projections of rank $d/2$ on $\mathbb C^d$. For each $x\in[g]$, we can define the observable $A_x = P_x - P_x^\perp = 2 P_x - I$, which is clearly such that $\|A_x\|_\infty = 1$. Choosing a scaling factor $t \in[0,1]$, we can then write
\begin{equation*}
    t A_x = 2 P_x^{(t)} - I, \text{ where } P_x^{(t)} = t P_x +(1-t) \frac{I}{2} \, .
\end{equation*}
We already know that if we choose $t \leq 1/\sqrt{g}$, then the noisy dichotomic PVMs $(\mathrm{P}_x^{(t)})_{x\in[g]}$ become compatible \cite{bluhm2018joint}. Our aim here is to show that this bound is typically not far from optimal. Concretely, we want to prove that for $t$ not too much larger than $1/\sqrt{g}$, the $(\mathrm{P}_x^{(t)})_{x\in[g]}$ are with high probability incompatible. This way, we could construct (at random) almost maximally incompatible dichotomic measurements.

In order to conclude that the $(\mathrm{P}_x^{(t)})_{x\in[g]}$ are incompatible, we need an incompatibility witness. As an Ansatz, we choose $W_x := s A_x/d$ for all $x \in [g]$. First, we need to verify for which values of $s$ the operator $W_x$ is indeed an incompatibility witness, i.e., for which $s$ it holds that $\langle W, A \rangle \leq 1$ for all tensors $A = \sum_{x=1}^g (2 E_x - I)\otimes\ket{x}$ that arise from compatible dichotomic measurements $\{E_1,\ldots,E_g\}$, where we have written $W = \sum_{x=1}^g W_x\otimes\ket{x}$ and $\{\ket{1}\!,\ldots,\ket{g}\}$ is the standard orthonormal basis of $\mathbb R^g$. From Eq.~\eqref{eq:incompatibility-witness-dichotomic}, we know that $W$ is an incompatibility witness if and only if there exists a quantum state $\rho$ such that
\begin{equation*}
    \rho - \sum_{x=1}^g \epsilon_x W_x \geq 0 \qquad \forall \epsilon \in \{\pm 1\}^g.
\end{equation*}
As an Ansatz, we choose $\rho= I/d$, and we thus obtain the following sufficient condition for $W$ to be an incompatibility witness:
\begin{equation*}
     \frac{I}{d} - \sum_{x=1}^g \epsilon_x W_x \geq 0 \qquad \forall \epsilon \in \{\pm 1\}^g,
\end{equation*}
i.e.~equivalently
\begin{equation*}
    \sum_{x=1}^g \epsilon_x A_x \leq \frac{I}{s} \qquad \forall \epsilon \in \{\pm 1\}^g.
\end{equation*}

Let us briefly comment on these Ansatz choices. Intuition from other areas of quantum information theory tells us that colinear witnesses typically perform well in detecting properties of high-dimensional random objects (such as e.g.~the entanglement of a random state, the non-locality of a random strategy, etc). As for the choice of the maximally mixed state, it is expected to get better and better as the number $g$ of measurements increases. Indeed, the independent random observables $A_1,\ldots,A_g$ are then expected to cover more and more uniformly all directions of the space of (trace $0$) Hermitian matrices, so that the best guess in order to upper bound their (signed) sum is a multiple of the identity matrix. 

\subsection{Infinite-dimensional results} 
\label{sec:more-proj-infinite}

Before proceeding, let us recall that for a $d \times d$ Hermitian matrix $M$, its \emph{empirical eigenvalue distribution} $\mu_M$ is defined as
\[
    \mu_M := \frac{1}{d} \sum_{i=1}^d \delta_{\lambda_i(M)},
\]
where $\lambda_i(M)$ are the eigenvalues of $M$ (counted with multiplicity). This distribution encodes the spectral properties of $M$ in the large $d$ limit.

It is known that the empirical eigenvalue distribution of each Hermitian matrix $A_x$, $x\in[g]$, as defined above converges (strongly almost surely) to a so-called (symmetric) Bernoulli distribution $b$ \cite{anderson2010introduction}, i.e.
\begin{equation*}
    \mu_{A_x} \xrightarrow[d \to \infty]{} \frac{1}{2}(\delta_{-1} + \delta_{1}) =: b.
\end{equation*}
And clearly, the same holds for $-A_x$. 

We are now going to use Voiculescu's result about the asymptotic freeness of unitarily invariant and independent random matrices. Voiculescu's groundbreaking work in free probability theory established that, in the large dimension limit, independent random matrices that are unitarily invariant (such as Haar-random unitary conjugates of fixed matrices) become \emph{asymptotically free} \cite{voiculescu1992free,nica2006lectures}. This means that, as the matrix size $d$ tends to infinity, the joint distribution of such matrices with respect to normalized trace behaves as if the matrices were free non-commutative random variables in the sense of free probability. In particular, the empirical spectral distribution of sums of these matrices converges to the \emph{free additive convolution} $\boxplus$ of their individual limiting spectral distributions. The $g$ random matrices $A_1,\ldots,A_g$ fall exactly in this situation, hence they are asymptotically free. It follows that, for any fixed $\epsilon\in\{\pm 1\}^g$, the spectrum of $A_\epsilon:=\sum_{x=1}^d \epsilon_x A_x$ converges to the free additive convolution of $g$ Bernoulli distributions, that is 
\begin{equation*}
    \mu_{A_\epsilon} \xrightarrow[d \to \infty]{} b^{\boxplus g} .
\end{equation*}
The distribution $b^{\boxplus g}$ is actually well known: it is a so-called Kesten-McKay distribution with parameter $g$ \cite{kesten1959symmetric}, whose density (see \cref{fig:Kesten-McKay}) is given by
\begin{equation*}
    f_{b^{\boxplus g}}(x) = \frac{1}{2\pi}\frac{g}{g^2-x^2}\sqrt{4(g-1)-x^2}\,\mathbf{1}_{|x|\leq 2\sqrt{g-1}} .
\end{equation*}
We thus have
\begin{equation*}
    \max \operatorname{supp} \left(b^{\boxplus g}\right) = 2 \sqrt{g -1} \, .
\end{equation*}

\begin{figure}[htb]
    \centering
    \includegraphics[width=0.4\linewidth]{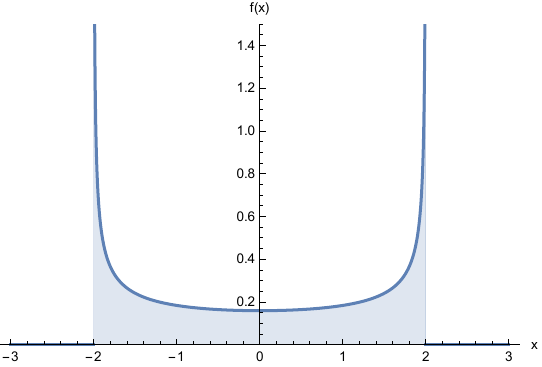} \qquad 
    \includegraphics[width=0.4\linewidth]{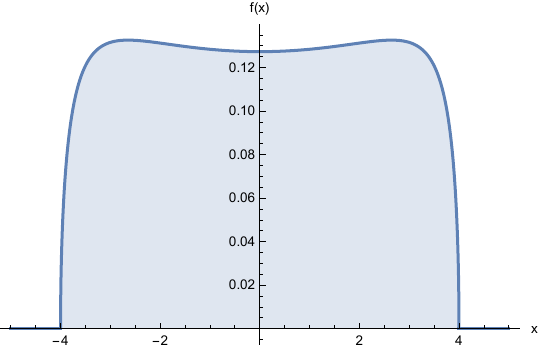}
    \caption{Kesten-McKay distribution for parameter values $g=2$ (left) and $g=5$ (right).}
    \label{fig:Kesten-McKay}
\end{figure}

As a consequence of the above discussion, we obtain the following: for the matrices $A_1,\ldots,A_g$ to form, asymptotically (i.e.~as $d\to\infty)$, an entanglement witness, we need that $ \max \operatorname{supp} (b^{\boxplus g}) \leq 1/s$ i.e.~$2\sqrt{g-1}\leq 1/s$. Since the convergence holds in the almost sure sense for all sign choices $\epsilon$, the optimal $s$ to choose is
\begin{equation*}
    s_{opt} = \frac{1}{2\sqrt{g -1}}.
\end{equation*}
Furthermore, we can compute that
\begin{equation}\label{eq:threshold-proj}
   \left \langle \frac{sA}{d}, t A \right\rangle = \frac{s t}{d} \sum_{x=1}^g  \operatorname{Tr}\left(A_x^2\right) = stg,
\end{equation}
where the last equality is because $A_x^2=I$ for each $x\in[g]$.
Hence, we find that the optimal $t$ to choose (see Section \ref{sec:incompatibility-witnesses}) is
\begin{equation*}
    t_{opt} = \frac{1}{gs_{opt}} = \frac{2\sqrt{g -1}}{g}.
\end{equation*}

We can summarize the observations above into the following proposition. 

\begin{proposition}
Let $P_1,\ldots,P_g$ be $g$ independent Haar-random projections of rank $d/2$ on $\mathbb C^d$. Then, the following holds almost surely:
$$ \limsup_{d \to \infty} \tau(P_1, \ldots, P_g) \leq \frac{2\sqrt{g-1}}{g}. $$
\end{proposition}

The above result tells us that the compatibility degree of $g$ independent Haar-random dichotomic balanced PVMs is, in the large $d$ limit, at most
\[ \frac{2\sqrt{g-1}}{g}  \underset{g\to\infty}{\sim} \frac{2}{\sqrt{g}}.
\]
Note that this upper bound is trivial in the case $g=2$ of two dichotomic balanced PVMs. However, this case has been discussed in great detail (and beyond the balanced regime) in \cref{sec:two-projections}. On the other hand, for large $g$ it differs by only a factor $2$ from the lower bound $1/\sqrt{g}$. We conjecture that this lower bound is actually the true value of the compatibility degree for all $g$ (see \cref{conj:g-proj-max-incomp} for a precise statement). 

\subsection{Finite-dimensional results} 
\label{sec:more-proj-finite}

Our goal here is to show that, for $d$ large but finite, if $t\geq C/\sqrt{g}$, then the $g$ independent Haar-random balanced dichotomic PVMs $\mathrm{P}_1^{(t)},\ldots,\mathrm{P}_g^{(t)}$ are with high probability incompatible (where $C<\infty$ is an absolute constant, independent of $d$ and $g$). We will follow the same strategy as before. Namely we will show that, if $s\leq c/\sqrt{g}$, for some absolute constant $c>0$, then with high probability, $\sum_{x=1}^g\varepsilon_xA_x\leq I/s$ for all $\varepsilon\in\{\pm 1\}^g$.

Define the Hermitian matrix $\Delta$ on $\C^d$ by
\[ \Delta = \sum_{i=1}^{d/2}\ketbra{i}{i}-\sum_{i=d/2+1}^{d}\ketbra{i}{i}. \]
We then have that, for each $x\in[g]$, $A_x=U_x^*\Delta U_x$, where $U_1,\ldots,U_g$ are independent Haar-distributed unitaries on $\C^d$.

\begin{lemma} \label{lem:moments}
Let $\varphi\in\C^d$ be a fixed unit vector and define the (real) random variable $X=\bra{\varphi}U^*\Delta U\ket{\varphi}$ for $U$ a Haar-distributed unitary on $\C^d$. Then, for all $p\in\N$, $\E[X^p]=0$ if $p$ is odd and $\E[X^p]\leq (p/2)!(2/d)^{p/2}$ if $p$ is even.
\end{lemma}

\begin{proof}
    Setting $\varphi_U=U\varphi\in\C^d$, we have $X=\Tr(\ketbra{\varphi_U}{\varphi_U}\Delta)$. Now, observe that $\varphi_U$ is a uniformly distributed unit vector in $\C^d$, and therefore 
    \[ \E\left[\ketbra{\varphi_U}{\varphi_U}^{\otimes p}\right] = \frac{1}{\binom{d+p-1}{p}}P_{S(d,p)} , \]
    where $P_{S(d,p)}$ is the orthogonal projector onto the symmetric subspace of $(\C^d)^{\otimes p}$, which has dimension $\binom{d+p-1}{p}$. We thus have
    \begin{align*}
       \E\left[X^p\right] & = \E\left[\Tr\left(\ketbra{\varphi_U}{\varphi_U}\Delta\right)^p\right] = \E\left[\Tr\left(\ketbra{\varphi_U}{\varphi_U}^{\otimes p}\Delta^{\otimes p}\right)\right] 
       = \frac{1}{\binom{d+p-1}{p}}\Tr\left(P_{S(d,p)}\Delta^{\otimes p}\right). 
    \end{align*}  
    Next, denoting by $\mathcal S_p$ the set of permutations of $p$ elements, we know that we can rewrite 
    \[ P_{S(d,p)} = \frac{1}{p!}\sum_{\pi\in\mathcal S_p}U_\pi, \]
    where for each $\pi\in\mathcal S_p$, $U_\pi$ stands for its associated unitary on $(\C^d)^{\otimes p}$ (which acts on product vectors as $U_\pi v_1\otimes\cdots\otimes v_p=v_{\pi^{-1}(1)}\otimes\cdots\otimes v_{\pi^{-1}(p)}$).
    It is then easy to check that, for each $\pi\in\mathcal S_p$,
    \[ \Tr\left(U_\pi\Delta^{\otimes p}\right) = \prod_{c\in C(\pi)}\Tr\left(\Delta^{|c|}\right), \]
    where $C(\pi)$ denotes the set of cycles of $\pi$ (in the decomposition of $\pi$ in disjoint cycles), and for each $c\in C(\pi)$, $|c|$ denotes the length of $c$. Now, for any $q\in\N$, $\Tr(\Delta^q)=0$ if $q$ is odd and $\Tr(\Delta^q)=d$ if $q$ is even. Hence, denoting by $|C(\pi)|$ the number of cycles of $\pi$, $\tr(U_\pi\Delta^{\otimes p})=d^{|C(\pi)|}$ if $\pi$ has only cycles of even length and $\tr(U_\pi\Delta^{\otimes p})=0$ otherwise. Denoting by $\mathcal S'_p$ the subset of $\mathcal S_p$ composed of elements whose cycles all have even length, we thus obtain
    \begin{equation} \label{eq:p-moment} \E\left[X^p\right] = \frac{1}{\binom{d+p-1}{p}} \frac{1}{p!}\sum_{\pi\in\mathcal S'_p}d^{|C(\pi)|} . \end{equation}
    Since $\mathcal S'_p=\emptyset$ when $p$ is odd, we immediately get from Eq.~\eqref{eq:p-moment} that $\E[X^p]=0$ when $p$ is odd. Suppose now that $p$ is even and write $p=2q$ with $q\in\N$. We can show recursively that
    \begin{equation} \label{eq:combinatorics} \sum_{\pi\in\mathcal S'_{2q}}d^{|C(\pi)|} = \prod_{k=1}^q(2k-1)(d+2(k-1)) = \frac{(2q)!}{q!}\frac{(d/2+q-1)!}{(d/2-1)!}, \end{equation}
    where the last equality is because, for any $m\in\N$, $(2m-1)!!=2^{-m}(2m)!/m!$ and $(2m)!!=2^mm!$. So let us now prove the first equality. Observe that $\pi\in\mathcal S_{2q}'$ can be constructed as follows: First, choose $\pi(1)$. There are $2q-1$ choices for that since it cannot be $1$. Then, identify $\{1,\ldots,2q\}\setminus\{1,\pi(1)\}$ with $\{1,\ldots,2(q-1)\}$ and pick $\hat\pi\in\mathcal S_{2(q-1)}'$. Finally, insert $(1,\pi(1))$ either inside a cycle of $\hat\pi$, after any element, or as a cycle on its own. In the first case $|C(\pi)|=|C(\hat\pi)|$, and there are $2(q-1)$ ways of doing that, while in the second case $|C(\pi)|=|C(\hat\pi)|+1$. Hence, setting $M_q(d)=\sum_{\pi\in\mathcal S'_{2q}}d^{|C(\pi)|}$, we have
    \[ M_q(d) = (2q-1)\sum_{\pi\in\mathcal S'_{2(q-1)}}\left((2(q-1)d^{|C(\pi)|}+d^{|C(\pi)|+1}\right) = (2q-1)\left(2(q-1)+d\right)M_{q-1}(d), \]
    from which the first equality in Eq.~\eqref{eq:combinatorics} follows, having in mind that $M_1(d)=d$. Inserting Eq.~\eqref{eq:combinatorics} into Eq.~\eqref{eq:p-moment} gives
    \begin{align*} \E\left[X^{2q}\right] & = \frac{\binom{d/2+q-1}{q}}{\binom{d+2q-1}{2q}} \\ 
    & = \frac{(2q)!}{q!}\frac{(d/2+q-1)!}{(d+2q-1)!}\frac{(d-1)!}{(d/2-1)!} \\
    & = \frac{2^q(2q-1)!!\times 2^{d/2-1}(d-1)!!}{2^{d/2+q-1}(d+2q-1)!!} \\
    & = \frac{(2q-1)(2q-3)\cdots 1}{(d+2q-1)(d+2q-3)\cdots(d+1)} \\
    & \leq \frac{2^qq!}{d^q},
    \end{align*}
    where the third equality is because, for any $k\in\N$, $(2k-1)!!=(2k)!/2^kk!=(2k-1)!/2^{k-1}(k-1)!$. And we get exactly the claimed result.
\end{proof}

Before we move on, let us make a very brief recap on sub-Gaussian random variables. We say that a (real) random variable $X$ is a (centered) sub-Gaussian random variable with variance proxy $\sigma^2$ if $\E[X]=0$ and its moment generating function satisfies
\begin{equation} \label{eq:sub-Gaussian} \forall\ s\in\mathbb R,\ \E\left[\exp(sX)\right] \leq \exp\left(\frac{\sigma^2s^2}{2}\right) . \end{equation}
Another equivalent way of defining sub-Gaussian random variables is through their moments. In particular, it can easily be shown that, if $X$ satisfies
\[ \forall\ q\in\mathbb N,\ \E\left[X^{2q}\right] \leq K^qq! , \]
then it satisfies Eq.~\eqref{eq:sub-Gaussian} with $\sigma^2=6K$ (see e.g.~\cite[Theorem 2.1.1]{pauwels2020stat} or \cite[Theorem 1.1.5]{chafai2012interactions}). The key property of sub-Gaussian random variables is that they have strong Gaussian-like concentration properties (see e.g.~\cite[Lemma 1.3]{rigollet2023high}). Concretely, if $X$ is a sub-Gaussian random variable with variance proxy $\sigma^2$, then
\[ \forall\ r>0,\ \P\left(X>r\right) \leq \exp\left(-\frac{r^2}{2\sigma^2}\right). \]
What is more, simply observing that a linear combination of independent sub-Gaussian random variables is itself sub-Gaussian, we get the following deviation inequality (see e.g.~\cite[Corollary 1.7]{rigollet2023high}): if $X_1,\ldots,X_g$ are independent sub-Gaussian random variable with variance proxy $\sigma^2$, then
\begin{equation} \label{eq:Bernstein} \forall\ r>0,\ \P\left(\frac{1}{g}\sum_{x=1}^g X_x>r\right)\leq \exp\left(-\frac{gr^2}{2\sigma^2}\right). \end{equation}

\begin{proposition} \label{prop:deviation-individual}
    Let $\varphi\in\C^d$ be a fixed unit vector and define the random variables $X_x=\bra{\varphi}U_x^*\Delta U_x\ket{\varphi}$, $x\in[g]$, for $U_1,\ldots,U_g$ independent Haar-distributed unitaries on $\C^d$. Then,
    \[ \forall\ r>0,\ \P\left(\sum_{x=1}^g X_x>r\sqrt{g}\right)\leq \exp\left(-\frac{dr^2}{24}\right). \]
\end{proposition}

\begin{proof}
    Clearly the random variables $X_1,\ldots,X_g$ are independent. What is more, by the discussion above, we see that Lemma \ref{lem:moments} implies that they are sub-Gaussian random variables with variance proxy $12/d$. As a consequence, we have by Eq.~\eqref{eq:Bernstein} that 
    \[ \forall\ r>0,\ \P\left(\frac{1}{g}\sum_{x=1}^g X_x>r\right)\leq \exp\left(-\frac{gdr^2}{24}\right). \]
    Applying this inequality with $r$ replaced by $r/\sqrt{g}$, we get the announced result.
\end{proof}

\begin{proposition} \label{prop:deviation-global}
    Let $U_1,\ldots,U_g$ be independent Haar-distributed unitaries on $\C^d$ and define the random Hermitian matrix $S=\sum_{x=1}^gU_x^*\Delta U_x$ on $\C^d$. Then,
    \[ \forall\ r>0,\ \P\left(\forall\ \varphi\in S_{\C^d},\ \bra{\varphi}S\ket{\varphi}\leq r\sqrt{g}\right)\geq 1- 144^{d}\exp\left(-\frac{dr^2}{96}\right). \]
\end{proposition}

\begin{proof}
    Fix $0<\delta<1/2$ and let $\mathcal N_\delta$ be a $\delta$-net in the Euclidean unit sphere $S_{\C^d}$ of $\mathbb C^d$, for the Euclidean norm. We know that we can choose it such that $|\mathcal N_\delta|\leq(3/\delta)^{2d}$ (see e.g.~\cite[Corollary 4.2.13]{vershynin2018high}). If this is so, we have by Proposition \ref{prop:deviation-individual}, together with the union bound, that
    \begin{equation} \label{eq:deviation-net} \forall\ r>0,\ \P\left(\forall\ \varphi\in\mathcal N_\delta,\ \bra{\varphi}S\ket{\varphi}\leq r\sqrt{g}\right)\geq 1- \left(\frac{3}{\delta}\right)^{2d}\exp\left(-\frac{dr^2}{24}\right). \end{equation}
    Now, we set $M=\max_{\varphi\in S_{\C^d}}\bra{\varphi}S\ket{\varphi}=\max_{\varphi,\varphi'\in S_{\C^d}}|\!\bra{\varphi}S\ket{\varphi'}\!|$ (where the last equality is because $S$ is Hermitian) and $M_\delta=\max_{\varphi\in\mathcal N_\delta}\bra{\varphi}S\ket{\varphi}$. It is easy to check that $M\leq M_\delta/(1-2\delta)$. Indeed, for any $\varphi\in S_{\C^d}$, let $\varphi'\in\mathcal N_\delta$ be such that $\|\varphi-\varphi'\|\leq\delta$. We then have
    \[ \bra{\varphi}S\ket{\varphi} = \bra{\varphi'}S\ket{\varphi'}+\bra{\varphi}S\ket{\varphi-\varphi'}+\bra{\varphi-\varphi'}S\ket{\varphi'} \leq M_\delta+2\delta M. \]
    And the conclusion follows by taking the supremum over $\varphi\in S_{\C^d}$ on the left hand side. This fact combined with Eq.~\eqref{eq:deviation-net} gives
    \begin{equation} \label{eq:deviation-sphere} \forall\ r>0,\ \P\left(\forall\ \varphi\in S_{\C^d},\ \bra{\varphi}S\ket{\varphi}\leq \frac{r\sqrt{g}}{1-2\delta}\right)\geq 1- \left(\frac{3}{\delta}\right)^{2d}\exp\left(-\frac{dr^2}{24}\right). \end{equation}
    Taking $\delta=1/4$ in Eq.~\eqref{eq:deviation-sphere}, and replacing $r$ by $2r$, gives the announced result.
\end{proof}

As an immediate consequence of Proposition \ref{prop:deviation-global} we have that, if $r$ is such that $r^2/96\geq 2\log(144)$, i.e.~$r\geq \sqrt{192\log(144)}$, then
\begin{equation} \label{eq:deviation-global}
    \P\left(\forall\ \varphi\in S_{\C^d},\ \bra{\varphi}S\ket{\varphi}\leq r\sqrt{g}\right)\geq 1- \exp\left(-\frac{dr^2}{192}\right).
\end{equation}

\begin{theorem} \label{th:deviation-final}
    Given $g\leq 3d$, let $U_1,\ldots,U_g$ be independent Haar-distributed unitaries on $\C^d$ and, for each $\varepsilon\in\{\pm 1\}^g$, define the random Hermitian matrix $S_\varepsilon=\sum_{x=1}^g\varepsilon_xU_x^*\Delta U_x$ on $\C^d$. Then,
    \[ \forall\ r\geq \sqrt{192\log(144)},\ \P\left(\forall\ \varepsilon\in\{\pm 1\}^g,\ S_\varepsilon\leq r\sqrt{g}I\right)\geq 1-\exp\left(-\frac{dr^2}{384}\right). \]
\end{theorem}

\begin{proof}
    First observe that, for any Hermitian matrix $M$ on $\C^d$ and any $\lambda>0$, $M\leq\lambda I$ is equivalent to, for all $\varphi\in S_{\C^d}$, $\bra{\varphi}M\ket{\varphi}\leq\lambda$. Next, for each $\varepsilon\in\{\pm 1\}^g$, $S_\epsilon$ has the same distribution as $S=\sum_{x=1}^gU_x^*\Delta U_x$. So by Proposition \ref{prop:deviation-global}, together with the union bound, we have that
    \[ \forall\ r\geq \sqrt{192\log(144)},\ \P\left(\forall\ \varepsilon\in\{\pm 1\}^g,\ S_\varepsilon\leq r\sqrt{g}I\right)\geq 1-2^g\exp\left(-\frac{dr^2}{192}\right). \]
    Now, if $r\geq \sqrt{192\log(144)}$, we have in particular that $dr^2/192\geq 2g\log(2)$ if $g\leq 3d$. And this implies precisely the announced result.
\end{proof}

Choosing $r = \sqrt{192\log(144)}$, we can straightforwardly derive from Theorem \ref{th:deviation-final} and Eq.\ \eqref{eq:threshold-proj} the result we were aiming at, namely: in the regime where $g\leq 3d$, if $t\geq 31/\sqrt{g}$, then the POVMs $(P_x^{(t)},I-P_x^{(t)})$, $x\in[g]$, are incompatible with probability larger than $1-e^{-2d}$. This is because $\sqrt{192\log(144)}=30.89\ldots\leq 31$ and $\log(144)/2=2.48\ldots\geq 2$

\begin{remark}
    We can actually prove a slightly more general version of Theorem \ref{th:deviation-final}, without putting a restriction on the respective scaling of $d$ and $g$. Concretely, we have
    \[ \P\left(\forall\ \varepsilon\in\{\pm 1\}^g,\ S_\varepsilon\leq C\max\left(\sqrt{g},\frac{g}{\sqrt{d}}\right)I\right)\geq 1-\exp\left(-c\max(d,g)\right), \]
    where $C<\infty$ and $c>0$ are absolute constants.
\end{remark}

\subsection{Comparison to previous work}
Previously, we knew that, for any $d,g\in\mathbb N$,
\begin{equation} \label{eq:lower-bounds-projections}
    \tau(d,g,(2, \ldots, 2)) \geq \max\left(c(d), \frac{1}{\sqrt{g}}\right)
\end{equation} 
where $c(d) = 4^{-d'} \binom{2d'}{d'}$, with $d'=\lfloor d/2 \rfloor$ (see points (1) and (3) of Proposition \ref{prop:bounds-incompat-degree}). On the other hand, the best known upper bound was 
\begin{equation*} 
    \tau(d,g,(2, \ldots, 2)) \leq \frac{1}{\sqrt{g}},
\end{equation*} 
which was valid only for $d \geq 2^{\lceil (g-1)/2\rceil}$ (see point (2) of Proposition \ref{prop:bounds-incompat-degree}). Our study of random projections provides us with an upper bound that has the same scaling as the lower bound for a much larger parameter range, namely
\begin{equation*}
    \tau(d,g,(2, \ldots, 2))  \leq \frac{31}{\sqrt{g}}
\end{equation*}
as soon as $d \geq g/3$ for $d$ even. 

Putting together this upper bound and the lower bound in Eq.~\eqref{eq:lower-bounds-projections}, using the observation that the minimum compatibility degree only goes down with increasing dimension, we thus actually get
\begin{equation*}
    \tau(d,g,(2, \ldots, 2)) = \Theta\left(\frac{1}{\sqrt{g}}\right) \qquad \forall g \leq 3d .
\end{equation*}

\section{Incompatibility of random basis measurements} \label{sec:random-bases}

We now move on to studying the incompatibility of $g\geq 2$ random basis measurements, using an incompatibility witness strategy similar to that of the previous section (but adapted to handle non-dichotomic measurements). We treat separately the cases $g=2$ and $g$ large (in the sense that it scales at least as the underlying dimension $d$). We leave open the intermediate regime of $g>2$ but fixed.

\subsection{General facts: incompatibility witnesses} 

We consider now the case of measurements in $g$ bases $U_1,\ldots,U_g$ of $\mathbb C^d$, i.e.~whose effects are of the form 
$$E_{i|x} = U_x \ketbra{i}{i} U_x^* \qquad \forall i \in [d], \forall x \in [g].$$
We would like to find a witness $W$ of the incompatibility of $E$. In order to do so, we look for $W$ colinear to $E$, i.e.~such that 
$$W_{i|x}=\alpha(U_1,\ldots,U_g)\, U_x \ketbra{i}{i} U_x^* \qquad \forall i \in [d], \forall x \in [g].$$
By Proposition \ref{prop:witness-bases}, we know that, setting 
\begin{equation}\label{eq:def-eta}
    \eta(U_1, \ldots, U_g):= \max_{f:[g] \to [d]} \lambda_{\max} \left( \sum_{x \in [g]} U_x \ketbra{f(x)}{f(x)} U_x^* \right),    
\end{equation}
the choice
\[ \alpha(U_1,\ldots,U_g) = \frac{1}{d\times\eta(U_1, \ldots, U_g)} \]
guarantees that $W$ is an incompatibility witness.

Let us re-write the quantity $\eta(U_1, \ldots, U_g)$ in a more tractable way. We have
\begin{align*}
    \eta(U_1, \ldots, U_g) & = \max_{f:[g] \to [d]} \max_{\|\phi\|=1} \sum_{x\in[g]} \langle \phi| U_x \ketbra{f(x)}{f(x)} U_x^* |\phi\rangle \\
    &= \max_{\|\phi\|=1} \max_{f:[g] \to [d]}  \sum_{x\in[g]} \langle f(x) |  U_x \ketbra{\phi}{\phi} U_x^* | f(x) \rangle \\
    &= \max_{\|\phi\|=1} \sum_{x\in[g]} \max_{i \in[d]} \langle i|  U_x \ketbra{\phi}{\phi} U_x^* | i \rangle \\
    &= \max_{\|\phi\|=1} \sum_{x\in[g]} \max_{i \in[d]} |(U_x \phi)_i|^2 \\
    &= \max_{\|\phi\|=1} \sum_{x\in[g]} \|U_x \phi\|_\infty^2 \\
\end{align*}

Using such witness $W$, we can prove that the noisy versions $E^{(t)}$ of the measurements are incompatible as soon as 
$$1 < \left\langle W, E^{(t)} \right\rangle = \frac{g}{\eta(U_1, \ldots, U_g)} \left(t + \frac{1-t}{d}\right)$$
i.e.~as soon as 
\begin{equation} \label{eq:incomp-eta}
    t > \frac{d\times\eta(U_1, \ldots, U_g)-g}{g(d-1)}.
\end{equation}

Clearly, $\eta(U_1, \ldots, U_g) \in [1,g]$, with 
$$\eta(U_1, \ldots, U_g) = g \iff \bigcap_{x \in [g]} \left\{ U_x^{-1} \ket i \, : \, i \in [d] \right\} \neq \emptyset.$$

\subsection{Incompatibility of two random basis measurements} 
\label{sec:random-bases-two}

Given unitary matrices $U_1,\ldots,U_g$, we clearly have 
$$\eta(U_1, \ldots, U_g) = \eta(U_1V, \ldots, U_gV)$$
for any unitary matrix $V$. We can thus always assume that one of the unitary matrices $U_1,\ldots,U_g$ is the identity. 

For $g=2$, using the first definition, we have 
$$\eta(I, U) = \max_{i,j \in [d]} \lambda_{\max} (\ketbra{i}{i} + U \ketbra{j}{j} U^*) = 1 + \max_{i,j \in [d]} |U_{ij}|.$$
Above, we have used the following fact: for two unit vectors $\phi,\psi$, $\lambda_{\max}(\ketbra{\phi}{\phi} + \ketbra{\psi}{\psi}) = 1 + |\langle\phi|\psi\rangle|$. This follows from writing, without loss of generality, $\ket{\phi} = \ket{1}$ and $\ket{\psi} = \cos(\theta)\!\ket{1}+ e^{i\alpha}\sin(\theta)\!\ket{2}$. Then, since $\Tr(\ketbra{\phi}{\phi} + \ketbra{\psi}{\psi}) = 2$ and $\det(\ketbra{\phi}{\phi} + \ketbra{\psi}{\psi}) = \sin^2\theta$, we have that the spectrum of $\ketbra{\phi}{\phi} + \ketbra{\psi}{\psi}$ is $1 \pm \cos \theta = 1 \pm |\langle\phi|\psi\rangle|$.

We shall now make use of the following result, providing the asymptotic behavior of the maximum amplitude of the entries of a Haar-distributed random unitary matrix. 

\begin{theorem}[{{\cite[Theorem 2]{jiang2005maxima}}}]
    Let $U^{(d)}$ be a sequence of Haar-distributed unitary $d\times d$ matrices. Then, almost surely, 
    \begin{align*}
        \liminf_{d \to \infty} \sqrt{\frac{d}{\log d}} \max_{i,j \in [d]} \left|U_{ij}^{(d)}\right| &= \sqrt 2, \\
        \limsup_{d \to \infty} \sqrt{\frac{d}{\log d}} \max_{i,j \in [d]} \left|U_{ij}^{(d)}\right| &= \sqrt 3.
    \end{align*}
\end{theorem}

As a straightforward consequence of this result, combined with Eq.~\eqref{eq:incomp-eta}, we get the following result.

\begin{corollary} \label{cor:two-bases}
    Let $U^{(d)}, V^{(d)}$ be sequences of unitary $d\times d$ matrices, such that $U^{(d)}, V^{(d)}$ are independent, and at least one of them is Haar-distributed. Then, for all sequences
    $$t_d > \frac 1 2\left(1 + \sqrt{\frac{3\log d}{d}}\right),$$
    almost surely as $d \to \infty$, the measurements
     $$\left( t_d U^{(d)}\ketbra{i}{i}{U^{(d)}}^* + (1-t_d)\frac{I_d}{d}\right)_{i\in[d]} \quad \text{ and } \quad \left( t_d V^{(d)}\ketbra{i}{i}{V^{(d)}}^* + (1-t_d)\frac{I_d}{d}\right)_{i\in[d]}$$
     are incompatible. 
\end{corollary}

\begin{remark}
    One can have a more precise reformulation of the result above by making use of \cite[Eq.~(2.21)]{jiang2005maxima}, which states that, for all $0<\epsilon<1$,
    $$\P\left[\sqrt{\frac{d}{\log d}} \max_{i,j \in [d]} \left|U_{ij}^{(d)}\right| \leq \sqrt 2(1-\epsilon) \right] = O\left(e^{-(\log d)^{3/2}}\right) \xrightarrow[d \to \infty]{} 1.$$
\end{remark}

\subsection{Incompatibility of many random basis measurements} 
\label{sec:random-bases-more}

We now consider the case where $g>2$. We want to get an estimate on $\eta(U_1, \ldots, U_g)$ for $U_1,\ldots,U_g$ independent Haar-distributed unitaries on $\C^d$; recall that the quantity $\eta$ was defined in \cref{eq:def-eta}. We start with the following simple observation: for any $\varphi\in\C^d$, we have 
\[ \|\varphi\|_\infty = \max\left\{ |\!\braket{\psi}{\varphi}\!| : \psi\in\C^d,\,\|\psi\|_1=1 \right\} = \max\left\{ |\!\braket{i}{\varphi}\!| : i\in[d] \right\}, \]
where the first equality is by duality between $\|\cdot\|_\infty$ and $\|\cdot\|_1$ while the second equality is because the extreme points of the unit ball for $\|\cdot\|_1$ are, up to a phase, $\{\ket{1}\!,\ldots,\ket{d}\}$. We can thus rewrite
\[ \eta(U_1, \ldots, U_g) = \max\left\{ \sum_{x=1}^g |\!\bra{i_x}U_x\ket{\varphi}\!|^2 : \varphi\in S_{\C^d},\,i_1,\ldots,i_g\in[d] \right\}.  \]

\begin{proposition} \label{prop:deviation-all-indiv-bases}
    Let $\varphi\in\C^d$ be a fixed unit vector and $i_1,\ldots,i_g\in[d]$ be fixed basis vectors. Define the random variables $X_x=|\!\bra{i_x}U_x\ket{\varphi}\!|^2$, $x\in[g]$, for $U_1,\ldots,U_g$ independent Haar-distributed unitaries on $\C^d$. Then,
    \[ \forall\ r>0,\ \P\left(\sum_{x=1}^g X_x>\frac{g}{d}(1+16r)\right) \leq \exp\left(-\frac{g}{2}\min(r^2,r)\right) . \]
\end{proposition}

Before proving Proposition \ref{prop:deviation-all-indiv-bases}, we need to recall one well-known fact about Gaussian random variables. Given $\gamma\in\mathbb C^d$ a (standard) Gaussian vector (i.e.~whose entries are independent mean $0$ and variance $1$ complex Gaussians), its Euclidean norm $\|\gamma\|$ and its direction $\varphi_\gamma=\gamma/\|\gamma\|$ are independent random variables. As a consequence, we have that, for any $i\in[d]$, 
\begin{equation} \label{eq:moments-gaussian} \forall\ q\in\mathbb N,\ \E\left[|\!\braket{i}{\varphi_\gamma}\!|^{2q}\right] = \frac{\E\left[|\!\braket{i}{\gamma}\!|^{2q}\right]}{\E\left[\|\gamma\|^{2q}\right]} \leq \frac{\E\left[|\!\braket{i}{\gamma}\!|^{2q}\right]}{\left(\E\left[\|\gamma\|^2\right]\right)^q} = \frac{\E\left[|\!\braket{i}{\gamma}\!|^{2q}\right]}{d^q}, \end{equation}
where the inequality is by Jensen inequality. Since $\braket{i}{\gamma}$ is a Gaussian random variable with variance $1$, this implies that $|\!\braket{i}{\varphi_\gamma}\!|$ is a sub-Gaussian random variable with variance proxy $1/d$ (where we have used the terminology introduced in Section \ref{sec:more-proj-finite}).

\begin{proof}
    For each $x\in[g]$, $\varphi_{x}=U_x\varphi$ is a uniformly distributed unit vector in $\C^d$, which means it has the same distribution as $\gamma_x/\|\gamma_x\|$ for $\gamma_x$ a Gaussian vector in $\C^d$. Hence, $Y_x=|\!\bra{i_x}U_x\ket{\varphi}\!|$ has the same distribution as $|\!\braket{i_x}{\gamma_x}\!|/\|\gamma_x\|$, which is a sub-Gaussian random variable with variance proxy $1/d$, as we have explained just above. This implies that $X_x-\E[X_x]=Y_x^2-\E[Y_x^2]$ is a so-called sub-exponential random variable with parameter $16/d$ (see e.g.~\cite[Definition 1.11]{rigollet2023high} and \cite[Lemma 1.12]{rigollet2023high}). Since $X_1,\ldots,X_g$ are additionally independent, we have by Bernstein inequality (see e.g.~\cite[Theorem 1.13]{rigollet2023high})
    \begin{equation} \label{eq:deviation-all-indiv-bases} \forall\ r>0,\ \P\left(\frac{1}{g}\sum_{x=1}^g \left(X_x-\E[X_x]\right) >r\right) \leq \exp\left(-\frac{g}{2}\min\left(\frac{d^2r^2}{256},\frac{dr}{16}\right)\right). \end{equation}
    We now just have to observe that, by Eq.~\eqref{eq:moments-gaussian} with $q=1$, for each $x\in[g]$, $\E[X_x]=1/d$. 
    Replacing $r$ by $16r/d$ in Eq.~\eqref{eq:deviation-all-indiv-bases} thus gives exactly the announced result.
\end{proof}

\begin{theorem} \label{th:deviation-final-bases}
    Given $g\geq 2d$, let $U_1,\ldots,U_g$ be independent Haar-distributed unitaries on $\C^d$. Then,
    \[ \forall\ s\geq 8,\ \P\left(\eta(U_1, \ldots, U_g)\leq \frac{g}{d}(7+16s\log(d)) \right) \geq 1-\exp\left(-\frac{sg\log(d)}{4}\right). \]
\end{theorem}

\begin{proof}
    First observe that, for any fixed unit vector $\varphi\in\C^d$, by Proposition \ref{prop:deviation-all-indiv-bases}, together with the union bound, we have
    \[ \forall\ r>0,\ \P\left(\forall\ i_1,\ldots,i_g\in[d],\ \sum_{x=1}^g |\!\bra{i_x}U_x\ket{\varphi}\!|^2\leq \frac{g}{d}(1+16r)\right) \geq 1- d^g\exp\left(-\frac{g}{2}\min(r^2,r)\right). \]
    As explained, this is equivalent to
    \begin{equation} \label{eq:deviation-indiv-bases} \forall\ r>0,\ \P\left(\sum_{x=1}^g \left\|U_x\varphi\right\|_\infty^2 \leq \frac{g}{d}(1+16r)\right) \geq 1- d^g\exp\left(-\frac{g}{2}\min(r^2,r)\right). \end{equation}
    
    Next, fix $0<\delta<1/2g$ and let $\mathcal N_\delta$ be a $\delta$-net in the Euclidean unit sphere $S_{\C^d}$ of $\mathbb C^d$, for the Euclidean norm. We know that that we can choose it such that $|\mathcal N_\delta|\leq(3/\delta)^{2d}$ (see e.g.~\cite[Corollary 4.2.13]{vershynin2018high}). If this is so, we have by Eq.~\eqref{eq:deviation-indiv-bases}, together with the union bound, that
    \begin{equation} \label{eq:deviation-net-bases} \forall\ r>0,\ \P\left(\forall\ \varphi\in\mathcal N_\delta,\ \sum_{x=1}^g \left\|U_x\varphi\right\|_\infty^2 \leq \frac{g}{d}(1+16r)\right) \geq 1- \left(\frac{3}{\delta}\right)^{2d}d^g\exp\left(-\frac{g}{2}\min(r^2,r)\right). \end{equation}
    Now, if $\varphi\in S_{\C^d}$ and $\varphi'\in\mathcal N_\delta$ are such that $\|\varphi-\varphi'\|\leq\delta$, then
    \begin{align*}
        \left| \sum_{x=1}^g \|U_x\varphi\|_\infty^2 - \sum_{x=1}^g \|U_x\varphi'\|_\infty^2 \right| & \leq \sum_{x=1}^g \left| \|U_x\varphi\|_\infty^2 - \|U_x\varphi'\|_\infty^2 \right| \\
        & \leq \sum_{x=1}^g \left( \|U_x\varphi\|_\infty + \|U_x\varphi'\|_\infty \right) \|U_x(\varphi-\varphi')\|_\infty \\
        & \leq \sum_{x=1}^g \left( \|U_x\varphi\| + \|U_x\varphi'\| \right) \|U_x(\varphi-\varphi')\| \\
        & = 2g\|\varphi-\varphi'\| \\
        & \leq 2g\delta,
    \end{align*}
    where the first inequality is by the triangle inequality, while the second inequality is by the triangle inequality as well, after writing $\|\psi\|_\infty^2-\|\psi'\|_\infty^2=(\|\psi\|_\infty+\|\psi'\|_\infty)(\|\psi\|_\infty-\|\psi'\|_\infty)$. This fact combined with Eq.~\eqref{eq:deviation-net-bases} gives 
    \[ \forall\ r>0,\ \P\left(\forall\ \varphi\in S_{\C^d},\ \sum_{x=1}^g \left\|U_x\varphi\right\|_\infty^2 \leq \frac{g}{d}(1+16r+2d\delta)\right) \geq 1- \left(\frac{3}{\delta}\right)^{2d}d^g\exp\left(-\frac{g}{2}\min(r^2,r)\right), \]
    which implies, taking $\delta=3/d$,
    \[ \forall\ r>0,\ \P\left(\forall\ \varphi\in S_{\C^d},\ \sum_{x=1}^g \left\|U_x\varphi\right\|_\infty^2 \leq \frac{g}{d}(7+16r)\right) \geq 1- d^{2d+g}\exp\left(-\frac{g}{2}\min(r^2,r)\right). \]
    If $g\geq 2d$, we get the announced result by taking $r=s\log(d)$ with $s\geq 8$ (which is such that $\min(r^2,r)=r$).
\end{proof}

Choosing $s=8$, we get the following from Theorem \ref{th:deviation-final-bases}: in the regime where $g \geq 2d$, with probability larger than $1-e^{-2g\log(d)}$, $\eta(U_1, \ldots, U_g)\leq 135\log(d)g/d$, which implies by Eq.~\eqref{eq:incomp-eta} that, for $t>(135\log(d)-1)/(d-1)$ (hence a fortiori for $t>135\log(d)/d$), the POVMs $(tU_x\ketbra{i}{i}U_x^*+(1-t)I/d)_{i\in[d]}$, $x\in[g]$, are incompatible.

\begin{remark}
    We can actually prove a slightly more general version of Theorem \ref{th:deviation-final-bases}, without putting a restriction on the respective scaling of $d$ and $g$. Concretely, we have
    \[ \P\left(\eta(U_1, \ldots, U_g)\leq C\log(d)\max\left(\frac{g}{d},1\right) \right) \geq 1-\exp\left(-c\max(g,d)\log(d)\right), \]
    where $C<\infty$ and $c>0$ are absolute constants.
\end{remark}

\subsection{Comparison to previous work} 

In terms of lower bounds on the compatibility degree, we know that, for any bases $\mathrm B_1,\ldots,\mathrm B_g$ of $\mathbb C^d$, 
\begin{equation} \label{eq:lower-bound-bases}
    \tau(\mathrm B_1,\ldots,\mathrm B_g) \geq \frac{g + d}{g(d+1)} .
\end{equation}
This result is obtained from cloning \cite{werner1998optimal}, and improves on the general lower bound recalled in point (5) of Proposition \ref{prop:bounds-incompat-degree} in the specific case of bases (see \cite{Heinosaari2014, bluhm2020compatibility}).

In the case $g=2$, the compatibility degree of $2$ MUBs $\mathrm B_1,\mathrm B_2$ of $\mathbb C^d$ has been exactly computed in \cite{Carmeli2012}, and is given by
\begin{equation*}
    \tau(\mathrm B_1,\mathrm B_2) = \frac{1}{2}\left(1 + \frac{1}{\sqrt{d} + 1} \right) .
\end{equation*}
Here, we have show that $2$ independent Haar-random bases of $\mathbb C^d$ have a compatibility degree that is asymptotically almost surely upper bounded by $1/2$, and hence actually equal to $1/2$ in light of the lower bound from Eq.~\eqref{eq:lower-bound-bases}. 

For $g>2$, we know from \cite{Carmeli2012, Designolle2018} (see also point (6) of Proposition \ref{prop:bounds-incompat-degree}) that $g$ MUBs $\mathrm B_1,\ldots,\mathrm B_g$ of $\mathbb C^d$ satisfy
\begin{equation*}
    \tau(\mathrm B_1,\ldots,\mathrm B_g) \leq \frac{g + \sqrt{d}}{g(\sqrt{d} + 1)} .
\end{equation*}
This upper bound is optimal for $g=2$, but the authors note that it is not very tight in general. Here we have studied the case of $g$ independent Haar-random bases $\mathrm B_1,\ldots,\mathrm B_g$ of $\mathbb C^d$. We have proved that, if $g\geq 2d$, then with high probability 
\begin{equation*}
     \tau(\mathrm B_1,\ldots,\mathrm B_g) \leq  \frac{135\log(d)}{d}.
\end{equation*}

Putting together this upper bound and the lower bound in Eq.~\eqref{eq:lower-bound-bases}, we get that, in the regime where $g$ is of order at least $d$, we have matching lower and upper bounds on $\tau(d,g,(d,\ldots,d))$ (up to log factors), namely
\[ \tau(d,g,(d,\ldots,d))= \tilde{\Theta}\left(\frac{1}{d}\right). \]

\section{Incompatibility of random induced measurements} \label{sec:random-POVMs}

Random quantum measurements are central to quantum information theory, both for foundational questions and for practical applications such as tomography and device-independent protocols. While much of the literature focuses on projective measurements (as does the current work up to here), realistic scenarios are more generally described by non-projective measurements. A particularly natural and physically motivated class of random POVMs arises when a quantum system interacts with an environment (ancilla), and a projective measurement is performed on the joint system. The effective measurement on the system alone is then described by a POVM, with the dimension of the ancilla controlling the amount of `randomness' or `noise' in the measurement. This model, known as the induced random POVM model, interpolates between projective measurements and maximally mixed measurements, and provides a rich landscape for studying measurement incompatibility.

Let us briefly recall the construction of this ensemble of random POVMs, as discussed in \cref{sec:random-matrix-models}. Consider a quantum system with Hilbert space $\mathbb{C}^d$ and an ancilla (environment) with Hilbert space $\mathbb{C}^n$. A projective measurement $(Q_i)_{i\in[k]}$ is performed on the joint space $\mathbb{C}^d \otimes \mathbb{C}^n$, where each $Q_i$ is a rank-$1$ projection. To obtain a random POVM on the system, we proceed as follows:
\begin{enumerate}
    \item Sample a random Haar-distributed isometry $V : \mathbb{C}^d \to \mathbb{C}^k \otimes \mathbb{C}^n$.
    \item Define the POVM $M$ on $\mathbb{C}^d$ with $k$ outcomes by
    \begin{equation*}
        M_i := V^*(\ketbra{i}{i} \otimes I_n) V \qquad \forall i \in [k].
    \end{equation*}
\end{enumerate}
This construction induces a probability measure $\nu_{d,k;n}$ on the set of $k$-outcome POVMs on a Hilbert space of dimension $d$. Alternatively, one can view the random POVM $M$ as the image of the standard basis measurement in $\mathbb{C}^k$ through a random unital, completely positive map. We refer the reader to \cite[Section V]{heinosaari2020random} for more details. 

The key parameter in this model is the ratio $c := nk/d$, which quantifies the relative size of the ancilla. The value $c = 1$ corresponds to a projective measurement, while $c \to 0$ corresponds to a trivial POVM $M_i = I/k$ for each $i\in[k]$. By varying $c$, one can interpolate between these regimes. Although $c$ plays a role analogous to the noise parameter $t$ in the white-noise models discussed previously (see e.g.~\cref{def:noise-model}), there is no exact correspondence between $c$ and $t$, and one should not directly compare their values or bounds.

In the previous sections, we considered PVMs, which are compatible if and only if the involved projections commute. For random projections, the measurements are generically incompatible, and we studied the robustness of incompatibility under added noise. In the induced random POVM scenario, the `noise' is controlled by the parameter $c$. The following proposition applies the incompatibility witness machinery developed in \cref{sec:incompatibility-witnesses} to independent induced random POVMs:

\begin{proposition}\label{prop:incompatibility-random-POVM}
    Fix $k,g\in\mathbb N$ and consider a sequence of $g$ independent random induced POVMs with $k$ outcomes, distributed along the measure $\nu_{d,k;n_d}$, in the regime where $d \sim c k n_d$ for some fixed parameter $c$ such that
    $$c > \frac{4(k-1)g}{(k-1)^2g^2 -2(k-2) (k-1)g+k^2}.$$
    Then, almost surely as $d \to \infty$, the $g$ POVMs are asymptotically incompatible. 
\end{proposition}

\begin{proof}
    Consider random POVMs $M_{i|x}$ as in the statement, where $i \in [k]$ and $x \in [g]$. We shall consider incompatibility witnesses of the form 
    $$W_{i|x} = \frac s d M_{i|x}$$
    for some specific value of $s$ to be determined. 

    First, consider an arbitrary function $f:[g] \to [k]$ and the associated operator 
    $$W_f := \sum_{x=1}^g W_{f(x) | x}.$$
    From the strong asymptotic freeness of Haar-random unitary matrices and deterministic matrices, we have that, almost surely,
    $$\lim_{d \to \infty} \lambda_{\max}(W_f) = s \times \max \operatorname{supp} \left(\nu_{k,c}^{\boxplus g}\right),$$
    where $\nu_{k,c}$ is the limit of the empirical eigenvalue distribution of an effect of a random POVM from the $\nu_{d,k;n_d}$ ensemble (see \cite[Proposition VI.2]{heinosaari2020random}), i.e.
    $$\begin{aligned}\nu_{k,c} :=
	D_c\left [b_{1/k}^{\boxplus 1/c}\right ] = \max\left(0, 1-\frac{1}{ck}\right)\delta_0 + \max\left(0, 1-\frac{1}{c}+\frac{1}{ck}\right)\delta_1  + \frac{\sqrt{(x-\varphi_-)(\varphi_+-x)}}{2 \pi c x (1-x)} \,\mathbf 1_{[\varphi_-,\varphi_+]}(x) \mathrm{d}x,
	\end{aligned}
    $$
	where $\varphi_\pm = \varphi_\pm(c,1/k)$ with 
     $$\varphi_\pm(s,t) := s + t - 2st \pm 2 \sqrt{st(1-s)(1-t)}.
	$$
    In particular, 
    \begin{align*}
        s \times \max \operatorname{supp} \left(\nu_{k,c}^{\boxplus g}\right) & = s \times \max \operatorname{supp} \left(D_c\left [b_{1/k}^{\boxplus g/c}\right]\right) \\
        & = s \times \max \operatorname{supp} \left(D_g \left[ D_{c/g} \left [b_{1/k}^{\boxplus g/c}\right] \right]\right) \\
        & = s \times \max \operatorname{supp} \left(D_g \left[\nu_{k,c/g}\right]\right) \\
        & = sg \times \varphi_+(c/g, 1/k),
    \end{align*}
    for all $g,k \geq 2$. 

    On the other hand, the quantity
    $$\langle W,M\rangle = \frac s d \sum_{x=1}^g \sum_{i=1}^k \Tr \left(M_{i|x}^2\right)$$
    converges, almost surely as $d \to \infty$ to 
    $$ sg\left(c + \frac{1-c}{k} \right),$$
    see \cite[Proposition VI.6]{heinosaari2020random}. 

    On the one hand, in order for the operators $W_{i|x}$ to be incompatibility witnesses one needs that
    $$ sg \times \varphi_+\left(\frac{c}{g},\frac{1}{k}\right) < 1.$$
    On the other hand, in order for these operators to witness the incompatibility of the POVMs $(M_{i|x})_{i\in[k]}$ one needs 
    $$ sg\left(c + \frac{1-c}{k} \right) > 1.$$
    One can find a real number $s$ satisfying both equations above if and only if the inequality in the statement holds. 
\end{proof}

This result establishes a threshold for incompatibility in terms of the parameter $c$, the number of outcomes $k$, and the number of measurements $g$. Physically, it means that as the ancilla dimension increases (i.e.~as $c$ increases), the measurements become `more random' and, above some threshold, generically incompatible. While as $c$ decreases, the measurements become trivial and thus compatible.

It is instructive to compare the incompatibility bound from the result above with other bounds for compatibility obtained in \cite{heinosaari2020random}. In \cref{fig:random-POVM-bounds}, we plot the curves corresponding to, on the one hand the incompatibility witness upper bound (blue curve), on the other hand the Jordan product criterion and the noise content criterion lower bounds (orange and green curves respectively). For $g=2$ random POVMs with $k$ outcomes, these criteria read, respectively: if
\begin{align*}
    c &< \frac{\left(3-2 \sqrt{2}\right) k+2 \left(\sqrt{2}-1\right)}{k^2+4 k-4} \qquad \textit{(Jordan product criterion)}\\
    c &< \frac{1}{6 k+4 \sqrt{(k-1) (2 k-1)}-4}  \qquad \textit{(noise content criterion)}
\end{align*}
then the two random POVMs are, almost surely as $d \to \infty$, compatible (see \cite[Corollaries VII.3 and VII.8]{heinosaari2020random} for the proofs). Note that, as the number of outcomes increases (i.e.~in the limit $k \to \infty$), the upper bound and the two lower bounds have the same $1/k$ scaling, with different constants: 
\begin{align*}
    c_{\text{witness}} &\sim \frac 8 k\\
    c_{\text{Jordan}} &\sim \frac{3-2\sqrt 2}{k} \approx \frac{0.171573} k\\
    c_{\text{noise}} &\sim \frac{1}{(6+4\sqrt 2)k} \approx \frac{0.0857864} k.    
\end{align*}

\begin{figure}
    \centering
    \includegraphics[width=0.5\linewidth]{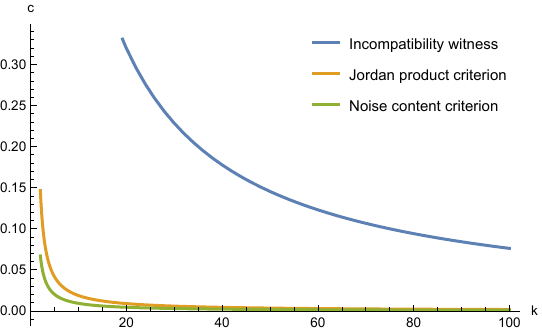}
    \caption{Upper and lower bounds for the compatibility of two independent induced random POVMs with $k$ outcomes and asymptotic parameter $c$. The top blue curve corresponds to the upper bound from \cref{prop:incompatibility-random-POVM}, while the two bottom curves correspond to the lower bounds coming from the Jordan product (orange) and the noise content (green) compatibility criteria. Parameters above the top curve correspond to asymptotically incompatible POVMs, while parameters below the orange curve correspond to asymptotically compatible POVMs. The compatibility of independent random induced POVMs having parameters between the two curves is undetermined.}
    \label{fig:random-POVM-bounds}
\end{figure}

The region between the upper (incompatibility) and lower (compatibility) bounds remains open: for parameters in this region, the asymptotic (in-)compatibility of random induced POVMs is not determined by current analytical criteria. Closing this gap is an interesting direction for future work, potentially requiring new analytical techniques or numerical investigations.

Going beyond two observables, the noise content criterion generalizes as follows \cite[Proposition 4 and Corollary 1]{heinosaari2013simple}: if
\begin{equation*}
    \sum_{x=1}^g \sum_{i=1}^k \lambda_{\min} \left(M_{i|x}\right) \geq g-1,
\end{equation*}
then the $g$ POVMs $(M_{i|x})_{i\in[k]}$ are compatible. This translates in our setting to the following result:

\begin{proposition} 
    Fix $k,g\in\mathbb N$ and consider a sequence of $g$ random induced POVMs with $k$ outcomes, distributed along the measure $\nu_{d,k;n_d}$, in the regime where $d \sim c k n_d$ for some fixed parameter $c$ such that 
    \begin{equation*}
        \varphi_-(c, 1/k) > \frac{g-1}{kg} \iff c > \frac{1}{g} \times \frac{1}{2 g (k-1)+2 \sqrt{(g-1) (k-1) (g (k-1)+1)}-k+2}.
    \end{equation*}
    Then, almost surely as $d \to \infty$, the $g$ POVMs are asymptotically compatible. 
\end{proposition}

Note that, in the above statement, we do not make any assumption on the joint distribution of the random induced POVMs: they could be independent or identical. Let us now compare how the incompatibility criterion from \cref{prop:incompatibility-random-POVM} compares with the compatibility criterion above, in the case of $g$ dichotomic ($k=2$) POVMs. We plot the curves in \cref{fig:random-POVM-inc-nc-k2}, and display the asymptotic (i.e.~as $g\to\infty$) threshold behaviors below:
\begin{equation*}
    c_{\text{witness}} \sim \frac{4}{g} \qquad c_{\text{noise}} \sim \frac{1}{4g^2}.
\end{equation*}
We see that the lower and upper bounds do not have the same scaling in $g$.

\begin{figure}
    \centering
    \includegraphics[width=0.5\linewidth]{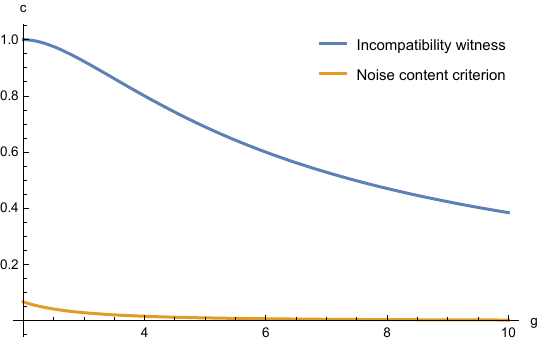}
    \caption{The upper bound given by the incompatibility criterion from \cref{prop:incompatibility-random-POVM} versus the lower bound given by the noise content criterion for $g$ dichotomic ($k=2$) random induced POVMs with asymptotic parameter $c$. The values of $c$ above the top blue curve correspond to asymptotically incompatible POVMs, while the values of $c$ below the bottom orange curve correspond to asymptotically compatible POVMs. The region between the two curves is undetermined.}
    \label{fig:random-POVM-inc-nc-k2}
\end{figure}

\section{Conclusion and outlook} \label{sec:discussion}

In this work, we have investigated the incompatibility properties of random quantum measurements across several models, including random dichotomic projective measurements, random basis measurements, and induced random POVMs. Our main finding is that, in high-dimensional regimes and for appropriate parameter choices, random measurements are \emph{typically almost maximally incompatible}: the amount of noise required to render them compatible is with high probability close to the theoretical maximum. This demonstrates that high incompatibility is not a rare or exceptional property, but rather a generic feature of quantum measurements sampled from natural random ensembles.

Beyond the probabilistic results, our analysis has led to the development and application of several non-random, structural techniques for certifying incompatibility. Notably, we have employed the method of \emph{incompatibility witnesses} (see Section~\ref{sec:incompatibility-witnesses}), which provides a powerful and flexible tool for certifying incompatibility via semidefinite programming. We have also introduced and exploited a \emph{compression technique} (see Section~\ref{sec:compression}), which allows us to relate the incompatibility of high-dimensional measurements to that of lower-dimensional ones, thereby transferring known results and simplifying the analysis.

Despite these advances, several open problems and directions for further research remain. A central question is whether the generalization of \cref{cor:two-proj-max-incomp} to $g>2$ measurements is true. More precisely, are $g$ independent Haar-random balanced dichotomic PVMs asymptotically close to maximally incompatible for all $g\geq 2$? We conjecture that the answer to this question is yes.
\begin{conjecture}\label{conj:g-proj-max-incomp}
    Given $g \geq 2$, let $P_1,\ldots,P_g$ be $g$ independent Haar-random projections of rank $d/2$ on $\mathbb C^d$ and denote by $A_1,\ldots,A_g$ their associated dichotomic observables. Then, the following holds in probability:
    \[ \lim_{d \to \infty} \tau(A_1, \ldots, A_g) = \frac{1}{\sqrt{g}}. \]
\end{conjecture}

Another natural question is whether the techniques developed here can be further improved or generalized. For instance, it would be particularly interesting to:
\begin{itemize}
    \item Apply the incompatibility witness technique to the case of two measurements, where current methods are less effective, and to develop witnesses that are optimal or nearly optimal in this regime.
    \item Extend the compression technique to settings involving more than two measurements, or to more general classes of measurements beyond those considered here; see \cref{conj:g-proj-max-incomp}.
    \item Adapt and refine these techniques to analyze incompatibility in other models of randomness, such as measurements with additional symmetries, constraints, or correlations.
    \item Use similar ideas and techniques to tackle other instances of non-classicality in quantum information theory, such as steering, marginal problems, entanglement, etc. 
\end{itemize}

\bigskip

\noindent\textbf{Acknowledgments.}  A.B.~was supported by the ANR project PraQPV, grant number ANR-24-CE47-3023. C.L.~and I.N.~were supported by the ANR project ESQuisses, grant number ANR-20-CE47-0014-01. We would like to thank Guillaume Aubrun for numerous helpful discussions throughout this project.

\bibliographystyle{alpha}
\bibliography{lit}
\end{document}